\documentclass[a4,12pt]{article}

\AtBeginDocument{%
  \providecommand\BibTeX{{%
    \normalfont B\kern-0.5em{\scshape i\kern-0.25em b}\kern-0.8em\TeX}}}

\usepackage[margin=0.7in]{geometry}

\usepackage{comment}
\usepackage{macros}
\mathchardef\mhyphen="2D 
\usepackage{algorithm,algorithmic}

\usepackage{amsmath, amssymb}
\usepackage{hyperref,color}
\hypersetup{
    colorlinks=true,   
    citecolor=blue,    
    linkcolor=blue,    
}

\usepackage{caption}
\usepackage{subcaption} 

\usepackage{authblk} 

\usepackage[hyperpageref]{backref}
\renewcommand*\backref[1]{\ifx#1\relax \else (Cited on #1) \fi}

\begin{document}

\title{An epistemic approach to model uncertainty in data-graphs}

\author[1,2]{Abriola Sergio}
\author[2]{Cifuentes Santiago}
\author[1,2]{Martínez María Vanina}
\author[1,2]{Pardal Nina}
\author[1]{Pin Edwin}
\affil[1]{ICC CONICET}
\affil[2]{Departamento de Computación, Universidad de Buenos Aires}

\date{}






\maketitle

\begin{abstract}
Graph databases are becoming widely successful as data models that allow to effectively represent and process complex relationships among various types of data. {\em Data-graphs} are particular types of
graph databases whose representation allows both data values in the paths and in the nodes
be treated as first class citizens by the query language. As with any other type of data repository, data-graphs may suffer from errors and discrepancies with respect to the real-world data they intend to represent.  In this work we explore the notion of {\em probabilistic unclean data-graphs}, in order to capture the idea that the observed (unclean) data-graph is actually the noisy version of a clean one that correctly models the world but that we know partially. As the factors that yield to such observation may be many, e.g, all different types of clerical errors or unintended transformations of the data, and depend heavily on
the application domain, we assume an epistemic probabilistic model that describes the distribution over all possible ways in which the clean (uncertain) data-graph could have been polluted. Based on this model we define two computational problems: \textit{data cleaning} and \textit{probabilistic query answering} and study for both of them their corresponding complexity when considering that the transformation of the data-graph can be caused by either removing (subset), adding (superset), or modifying (update) nodes and edges. For data cleaning, we explore restricted versions when the transformation only involves updating data-values on the nodes. Finally, we look at some implications of  incorporating hard and soft constraints to our framework.
\end{abstract}







\section{Introduction}

There is an increasing interest on graph databases as a mean to adequately handle both the topology of the data and the data itself, which is specially useful for applications that involve analysis over linked and  semi-structured data \cite{anand2010techniques,arenas2011querying,fan2012graph}. In these settings, the structure of the database is queried through navigational languages such as \textit{regular path queries} or RPQs~\cite{barcelo2013querying} that can capture pairs of nodes connected by a specific kind of path. More expressive query languages have been defined with a tradeoff on evaluation complexity, nevertheless,  RPQs and its most common extensions (C2RPQs and NREs ~\cite{barcelo2012relative}) do not treat data values in the nodes of the graph as first class citizens in the language leaving behind any possible interaction with that source of data. For this reason, query languages have been defined for the case of data-graphs (i.e. graph databases where data lies both in the paths and in the nodes themselves), such as REMs and \Gregxpath~\cite{libkin2016querying}.  

When accessing or querying a data repository we expect to obtain data that comply, at least to a certain extent, to the semantics of the domain, either in terms of quality or by satisfying a series of integrity constraints. This is a major challenge for any non-trivial data-driven application. Specifically for data-graphs, integrity constraints can be expressed in graph databases through \textit{path constraints}~\cite{abiteboul1999regular,buneman2000path}.

In the literature, two main approaches have been thoroughly explored. On the one hand, data cleaning, or data repairing, focuses on frameworks that allow to identify inconsistencies caused by incorrect, missing, and duplicated data, and to restore it to a {\em clean} state, i.e., a state that satisfies the imposed constraints. On the other hand, repairing the data may not always be the best option, or it may actually be impossible, e.g, we may not have the permission to actually change it. In these cases, a different approach is to develop the means (theory and algorithms) to obtain consistent answers from inconsistent databases without changing the data itself. The field of consistent query answering (CQA), first defined for relational databases~\cite{arenas1999consistent, lian2010consistent} has lately been applied to semi-structured data such as graph databases~\cite{barcelo2017data}. 

One of the main drawbacks of traditional approaches to repairing and consistent query answering is that they do not allow to represent the fact that, in general, when we observe the {\em unclean} data we may not know exactly how the data was corrupted as a variety of factors may have been involved.
A probabilistic epistemic model allows incorporating uncertain domain knowledge to the framework, formalizing both a set of possible original {\em clean} databases and an application-dependent noisy process that yields for each one potentially distorted version. 
For this reason, inspired by the work in~\cite{de2018formal, rekatsinas2017holoclean}, in this paper we propose a probabilistic framework for repairing and querying data-graphs, assuming an epistemic model that describes the distribution over all possible ways in which the clean (uncertain) data-graph could have been polluted. 
We adapt the definitions from~\cite{de2018formal} to the context of data-graphs, especially focusing on complexity aspects, and capitalizing on graph-related logics to define natural constraints and problems. We also take advantage of common theoretical devices developed for studying the complexity of reasoning problems related to database repairing. Overall, we aim to find tractable versions of some of the problems presented for the PUD framework~\cite{de2018formal} (this is, \textit{data cleaning} and \textit{PQA}), in which usual data-graph reasoning problems can be modelled. In particular, we concentrate in use cases where data is altered at the node level.
In this work particularly, we study probabilistic data cleaning and query answering for data-graphs focusing on three types of transformations that can be applied to the data-graph: by either removing or adding nodes and edges, and by updating data values in the nodes.
As far as we know from literature review,~\cite{de2018formal} is the only work that deals with issues of complexity for a probabilistic framework such as the one presented in this paper. However, in the aforementioned work, only some initial complexity considerations are explored and there are no results on intractability or that yield higher than polynomial complexity. Moreover, as we are dealing with data-graphs, we can more naturally express and tackle problems involving some kinds of constraints --such as path constraints-- that are more convenient for graph related applications. 
The specific contributions of this work are as follows:
\begin{itemize}

\item We define the notion of a probabilistic unclean data-graph model (PUDG) based on the observed data-graph and an epistemic model of the data-graph (EMDG); the latter is composed of a probabilistic distribution over all possible clean data-graphs and a realization model that represents the noisy process which allows us to ‘observe’ a potentially distorted version of each of the clean versions.
\item We restrict the (epistemic model) EMDG to subset EMDG and superset EMDG considering realization models that either only delete or only add data in the data-graphs, respectively. These notions relate to two common semantics for repairing databases, namely subset and superset repairs~\cite{tenCate:2012}. 
We also consider the Node-Update EMDG version, where only modifications to data-values in the observed data-graph are allowed. 

\item We study the problem of {\em data cleaning}, this is, given a probabilistic unclean data-graph 
$\aPUDG = (\aProbabilisticDatabase,\aRealizationModel,\aGraph')$, find the most probable data-graph in $\aProbabilisticDatabase$ given an epistemic model of the application domain and an observation. We explore different restrictions of the problem, trying to understand the range of complexity that can be involved when reasoning around these structures.

\item We consider alternate ways of defining a system of priorities by adding constraints to the model. We observe that the concepts of hard and soft integrity constraints~\cite{arenas1999consistent, barcelo2017data, tenCate:2012} can be easily incorporated in the framework.

\item Finally, we study the problem of {\em probabilistic query answering} (PQA), which computes the probability of a formula $\aFormulaNodeOrPath$ over $\aProbabilisticDatabase$, given a  probabilistic unclean data-graph $\aPUDG = (\aProbabilisticDatabase,\aRealizationModel,\aGraph')$.

\end{itemize}

This work is organized as follows. First, in Section~\ref{sec:RW} we discuss related work to provide some background to the proposal. In Section~\ref{Section:Definitions} we introduce the necessary preliminaries and notation for the syntax and semantics for our probabilistic model (PUDG). In Section~\ref{section:dataCleaning} we define the problem of probabilistic data cleaning and study its complexity for Subset, Superset, and Update PUDGs. We finish the section by considering restricted versions of Node-Update PUDG by applying different kinds of cardinality constraints to this framework.
In Section~\ref{section:constraints} we briefly discuss the generalization of our framework to
consider soft and hard constraints, and study the complexity of solving the Data Cleaning problem in some particular cases of hard constraints.
Later, in Section~\ref{section:PQA} we study the problem of probabilistic query answering for PUDGs focusing on  the subset and superset versions of the epistemic model.
Finally, in Section~\ref{Section:Conclusions} we discuss some final remarks and possible continuation of this work.






\section{Related Work}\label{sec:RW}

Probabilistic databases have been studied over the last 20 years by both the Database and the Artificial Intelligence community~\cite{van2017query}. These studies were motivated by a variety of applications, such as database repairing~\cite{rekatsinas2017holoclean}, modeling uncertain data~\cite{sarma2009representing}, data cleaning~\cite{de2018formal} and approximate query processing~\cite{zeng2014analytical}. Nonetheless, the main purpose of probabilistic databases is to extend today’s database technology to handle uncertain data while avoiding developing a new artifact from scratch~\cite{van2017query}. Many of the usual database techniques, such as query optimization and indexes, can be carried over (with some necessary adjustments) to a probabilistic database to follow the new probabilistic semantics of the language. 

The main problem that was studied in this context is query evaluation~\cite{van2017query, dalvi2007efficient,olteanu2008using}: given a query $Q$ and a probabilistic database $D$, compute the answer of $Q$ over $D$. This problem is known as \textit{probabilistic query answering} or \textit{probabilistic query evaluation} (PQE), which frames the task of 
finding the probability associated to each answer to $Q$ that is yielded by $D$. Though several semantics have been defined, overall, the task to compute the probability of $Q$ over $D$ can be reduced to enumerating every possible world $d$ (a non-probabilistic database) that satisfies $Q$, assigning to that answer a weight that is related to the probability of $d$ given the semantics. Through this interpretation, it can be seen that the PQE problem is similar to the \textit{weighted model counting problem}~\cite{beame2013model, gribkoff2014understanding}, and it is the case that tight bounds on the complexity of the PQE problem can be deduced using techniques from the latter one. An important dichotomy result of the area is that, for a fixed query $Q$, PQE is either a $\#P-complete$ problem (as usual for problems related to counting satisfying assignments of formulas) or it can be solved in polynomial time~\cite{dalvi2013dichotomy}.

The probabilistic database $D$ is commonly restricted to be a \textit{tuple-independent database} (TID), in which every tuple is an independent probabilistic event with its own associated value. Even though more expressive alternatives were proposed, such as block-disjoint-independent databases~\cite{dalvi2007management} or seeing attributes as random variables~\cite{arenas2020counting}, TIDs are the best understood alternative to date, and are already being used in 
applications such as relational embeddings~\cite{friedman2020symbolic}. In some cases, the probabilistic database $D$ is restricted further by requesting it to be \textit{symmetric}: any two tuples of the same relation are conditioned to have the same probability. When considering this special case, the PQE problem can be solved in polynomial time for a larger class of queries~\cite{van2014skolemization}.

A problem that is quite similar to the \textit{data cleaning} version we study in this work is the problem of computing the \textit{most probable database} (MPD), studied in~\cite{gribkoff2014most}. In the MPD problem, the task is to find the most probable deterministic database $d$ that satisfies a given query $Q$ over a probabilistic database $D$. Even when the probabilistic database $D$ is assumed to be a TID and the query $Q$ some kind of key or functional dependency, the problem becomes intractable for quite simple queries.

Probabilistic models have already been developed in other contexts aside from the relational one, such as with XML~\cite{kimelfeld2009modeling,souihli2013optimizing,abiteboul2011capturing}, ontologies~\cite{riguzzi2015reasoning}, and graphs~\cite{lian2010consistent}. In the case of graphs, most of the attention has been placed in the problem of querying a probabilistic graph database~\cite{lian2010consistent,amarilli2017conjunctive,maniu2017indexing}, where the underlying distribution over the state of the database is defined through edge independent probabilities, in an analogous way as in TIDs. 
Queries usually consist of some kind of path or graph pattern, and the satisfiability of a query is defined in terms of the probability that the complete pattern can be found in the graph based on the independent probabilities of the edges. As far as we know, there is no work on modelling the graph databases considering data in the nodes (in the same way as the XML trees do) and allowing the graph patterns to interact with them.


\section{Definitions}\label{Section:Definitions}

Fix a finite set of edge labels $\Sigma_e$ and a countable set (either finite or infinite enumerable) of data values $\Sigma_n$ (sometimes called data labels), which we assume non-empty, and such that $\Sigma_e \cap \Sigma_n = \emptyset$. A (finite) \defstyle{data-graph} $\aGraph$ is a tuple $(V_\aGraph,L_\aGraph,\dataFunction_\aGraph)$, or just $(V,L,\dataFunction)$ if $\aGraph$ is clear, where $V$ is a finite set of natural numbers, $L$ is a mapping from $V \times V$ to $\parts{\Sigma_e}$ defining the edges of the graph, and $\dataFunction$ is a function mapping the nodes from $V$ to data values in $\Sigma_n$. We denote by $E_\aGraph$ the set $\{(v,e,w)$ $|$ $v,w \in V, e \in L(v,w)\}$, which can be understood as the set of edges of the graph. Also, we denote by $G \subseteq G'$ when $V_G \subseteq V_{G'}$, $E_G \subseteq E_{G'}$, and $D_G(x) = D_{G'}(x)$ for every $x \in V_G$.

See Figure~\ref{figure:exampleDataGraph} for a visual example of a data-graph.

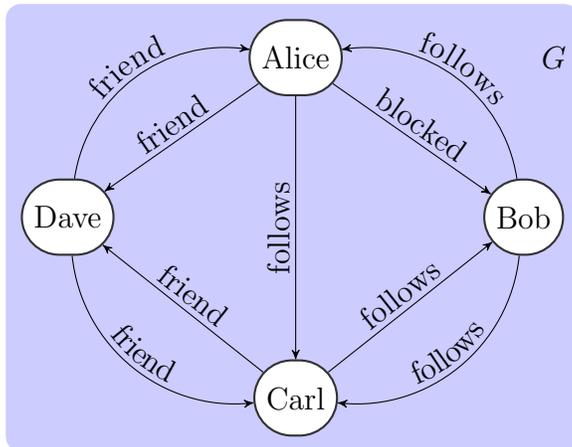
\begin{figure}[h!]
\centering
\begin{tikzpicture}[node distance=3.5cm,->,>=stealth', bend angle=45,auto] 
  \tikzstyle{persona}=[rounded rectangle,thick,draw=black!80,fill=black!0,minimum size=10mm]
  \tikzstyle{every label}=[blue]
  \tikzstyle{every decoration}=[text format delimiters={|}{|}] 

    \node [persona] (Alice) {Alice};
    \node [draw = none] (graph_clean_name) [right= 2.5cm of Alice] {$\aGraph$};
    \node [persona] (Bob) at ($ (Alice) + (-35:3.7cm) $) {Bob};
    \node [persona] (Carl) [below = of Alice] {Carl};
    \node [persona] (Dave) at ($ (Alice) + (-145:3.7cm) $) {Dave};
    
    \node [draw = none] (Path_start) [right=0.2cm of Bob] {};

    \draw [->, postaction= {decorate,decoration={text along path,text align=center, text={blocked}, raise=0.2em}}] (Alice) to (Bob) node {}; 
   \draw [->, postaction= {decorate,decoration={text along path,text align=center,text={follows},raise=0.2em, reverse path}}, bend right] (Bob) to (Alice) node {};
   \draw [->, postaction= {decorate,decoration={text along path,text align=center,text={follows},raise=0.2em, reverse path}}] (Alice) to (Carl) node {};   
    \draw [->, postaction= {decorate,decoration={text along path,text align=center,text={friend},raise=0.2em, reverse path}}] (Alice) to (Dave) node {};
    \draw [->, postaction= {decorate,decoration={text along path,text align=center,text={friend},raise=0.2em}}, bend left] (Dave) to (Alice) node {};
    \draw [->, postaction= {decorate,decoration={text along path,text align=center,text={ friend},raise=0.2em}}, bend right] (Dave) to (Carl) node {};
    \draw [->, postaction= {decorate,decoration={text along path,text align=center,text={ friend},raise=0.2em, reverse path}}] (Carl) to (Dave) node {};  
    \draw [->, postaction= {decorate,decoration={text along path,text align=center,text={follows},raise=0.2em}}] (Carl) to (Bob) node {};    
    \draw [->, postaction= {decorate,decoration={text along path,text align=center,text={ follows},raise=0.2em, reverse path}}, bend left] (Bob) to (Carl) node {};    

  \begin{pgfonlayer}{background}
    \filldraw [line width=4mm,join=round,blue!20]
      (Alice.north  -| Bob.east)  rectangle (Carl.south  -| Dave.west);   
  \end{pgfonlayer}

\end{tikzpicture}
\caption{A data-graph $\aGraph$ with set of nodes $V = \{1,2,3,4\}$  and data values $D(1) = \esDato{Alice}$, $D(2) = \esDato{Dave}$, $D(3) = \esDato{Carl}$ and $D(4) = \esDato{Bob}$, representing the names of the individuals in a social network, and the relation between any two individuals are either \esLabel{friend}, \esLabel{follows}, or \esLabel{blocked}.}  
\label{figure:exampleDataGraph}
\end{figure} 

\begin{definition}
A \defstyle{probabilistic data-graph} $\aProbabilisticDatabase$ is a probability distribution over a set of data-graphs.
More precisely, we consider 
$\aProbabilisticDatabase: \aUniverseOfGraphs \to [0,1]$, where $\aUniverseOfGraphs$ is the set of all data-graphs over $\Sigma_e, \Sigma_n$, and $$\sum_{\aGraph \in \aUniverseOfGraphs}\aProbabilisticDatabase(\aGraph) = 1.$$

For simplicity, we will sometimes abuse the notation and write $\aGraph \in \aProbabilisticDatabase$, or $\aGraph$ in $\aProbabilisticDatabase$, to refer to the case where $\aGraph$ is in the support of $\aProbabilisticDatabase$ (i.e., when  $\aProbabilisticDatabase(\aGraph)>0$).
\end{definition}



The probability distribution $\aProbabilisticDatabase$ can represent different types of uncertainty, such as the details of a concrete database, our priors about the relative frequency of different databases within a certain class, etc. The value $\aProbabilisticDatabase(\aGraph)$ is the probability that the original state of an uncertain data-graph is precisely $\aGraph$. 




\begin{example} 
As a toy example, we can consider a probabilistic data-graph that represents our epistemic state about the social relationship between four particular persons. 

Consider then $\Sigma_e = \{\esLabel{friend}\}$ and $\Sigma_n$ a (possibly infinite) set of valid names.
Let $V = \{1,2,3,4\}$, where each element of $V$ represents a different person, and $\aUniverseOfGraphs$ consists of data-graphs over $\Sigma_e, \Sigma_n, V$.
Now, assume that we are certain about the name of the four individuals to be $\{a_1, a_2, a_3, a_4\}$. Furthermore, assume that we know that $1$ and $4$ are mutual friends, but we are completely ignorant about other relationships among the four persons. Then, we could represent our epistemic state as the single probabilistic data-graph $\aProbabilisticDatabase: \aUniverseOfGraphs \to [0,1]$ where all the following conditions hold: 
\begin{itemize}
\item $\aProbabilisticDatabase(\aGraph) = 0$ for any $\aGraph$ such that, for some $ 1 \le i\le 4$, $D(i) \neq a_i$.
\item $\aProbabilisticDatabase(\aGraph) = 0$ for any $\aGraph$ such that either $(1, \esLabel{friends}, 4) \not \in E_\aGraph$ or $(4, \esLabel{friend}, 1) \not \in E_\aGraph$.
\item $\aProbabilisticDatabase(\aGraph) = 0$ for any $\aGraph$ such that $(i, \esLabel{friend}, i) \in E_\aGraph$ for some $i$.
\item $\aProbabilisticDatabase$ assigns the same probability to any $\aGraph$ that is not in the previous conditions.
\end{itemize}

The first condition ensures that the names of the individuals are precisely those that we know for certain. The second condition enforces that $1$ and $4$ are mutual friends, as requested. In other words, any graph with positive probability needs to have an edge from $1$ to $4$ with label $\esLabel{friend}$ and analogously from $4$ to $1$. As for the third condition, we ask that no individual is a friend of itself, this is, that every data-graph that represents this social relationship has no loops. Finally, the last condition states that those graphs that satisfy the first three constraints are equiprobable, disregarding of which other relationships they model.

In particular, it is easy to see that there are exactly $2^{2*3 + 2*2}$ data-graphs that fulfil these conditions, and thus they have assigned a non-negative probability, which is precisely $\frac{1}{1024}$. See Figure~\ref{Figure:AProbabilisticDataGraphToyExample} for an illustration.



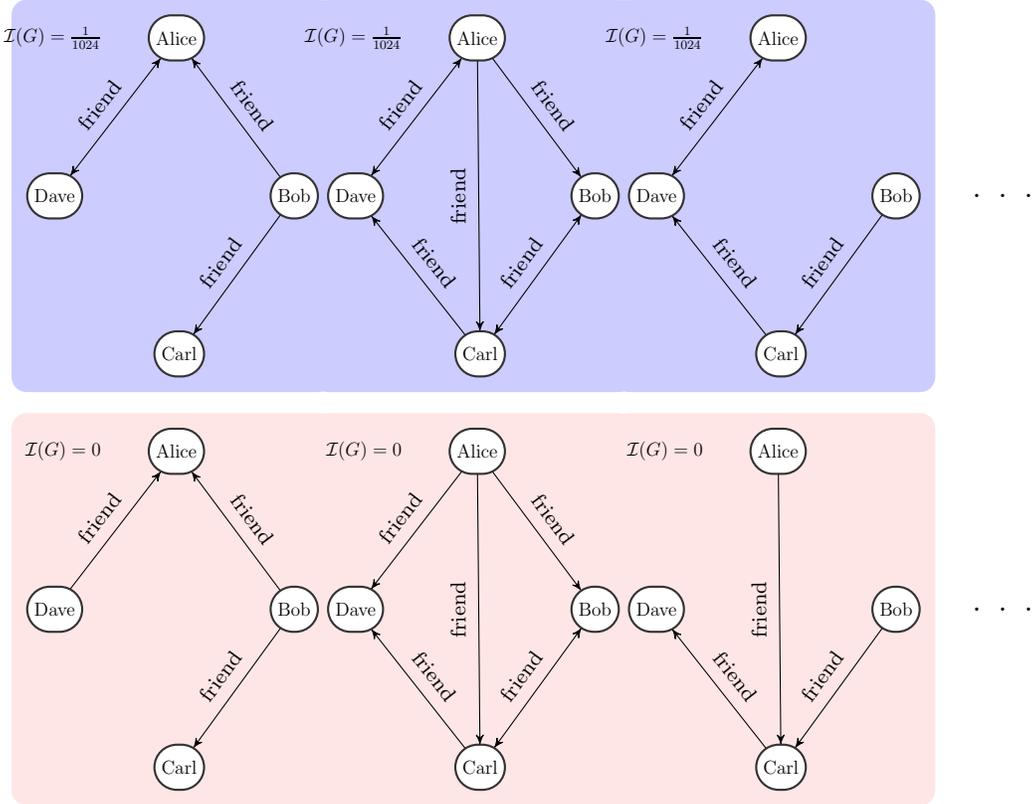
\begin{figure}[ht!]%
\centering 
\begin{tikzpicture}[node distance=3.5cm,->,>=stealth', bend angle=45,auto] 
  \tikzstyle{persona}=[rounded rectangle,thick,draw=black!80,fill=black!0,minimum size=10mm]
  \tikzstyle{every label}=[blue]
  \begin{scope}[xshift=-4cm, yshift=0cm, scale=0.6, every node/.append style={transform shape}]
    \node [persona] (Alice) {Alice};
    \node [draw = none] (graph_clean_name) [left= 0.9cm of Alice] {$\aProbabilisticDatabase(\aGraph) = \frac{1}{1024}$};
    \node [persona] (Bob) [below right = of Alice] {Bob};
    \node [persona] (Carl) [below left = of Bob] {Carl};
    \node [persona] (Dave) [below left = of Alice] {Dave};
    
    \node [draw = none] (Path_start) [right=0.2cm of Bob] {};
    \node [draw = none] (Path_start_up) [above =1cm of Path_start] {};
    
     \draw [<->, postaction= {decorate,decoration={text along path,text align=center,text={|\scriptsize|friend},raise=5pt}}] (Dave) to (Alice) node {}; 
    
   \draw [->, postaction= {decorate,decoration={text along path,text align=center,text={|\scriptsize|friend},raise=5pt, reverse path}}] (Bob) to (Alice) node {};
    \draw [->, postaction= {decorate,decoration={text along path,text align=center,text={|\scriptsize|friend},raise=5pt, reverse path}}] (Bob) to (Carl) node {}; 
  \end{scope}

  \begin{scope}[xshift=0cm, yshift=0cm, scale=0.6, every node/.append style={transform shape}]
    \node [persona] (Alice2) {Alice};
    \node [draw = none] (graph_clean_name) [left= 0.9cm of Alice2] {$\aProbabilisticDatabase(\aGraph) = \frac{1}{1024}$};
    \node [persona] (Bob2) [below right = of Alice2] {Bob};
    \node [persona] (Carl2) [below left = of Bob2] {Carl};
    \node [persona] (Dave2) [below left = of Alice2] {Dave};
    
    \node [draw = none] (Path_start2) [right=0.2cm of Bob2] {};
    \node [draw = none] (Path_start_up2) [above =1cm of Path_start2] {};
    
     \draw [<->, postaction= {decorate,decoration={text along path,text align=center,text={|\scriptsize|friend},raise=5pt}}] (Dave2) to (Alice2) node {}; 
    
   \draw [->, postaction= {decorate,decoration={text along path,text align=center,text={|\scriptsize|friend},raise=5pt}}] (Alice2) to (Bob2) node {};
   \draw [->, postaction= {decorate,decoration={text along path,text align=center,text={|\scriptsize|friend},raise=5pt, reverse path}}] (Alice2) to (Carl2) node {};
   \draw [->, postaction= {decorate,decoration={text along path,text align=center,text={|\scriptsize|friend},raise=5pt, reverse path}}] (Carl2) to (Dave2) node {};
    \draw [<->, postaction= {decorate,decoration={text along path,text align=center,text={|\scriptsize|friend},raise=5pt, reverse path}}] (Bob2) to (Carl2) node {}; 
  \end{scope}
  
  \begin{scope}[xshift=4cm, yshift=0cm, scale=0.6, every node/.append style={transform shape}]
    \node [persona] (Alice3) {Alice};
    \node [draw = none] (graph_clean_name) [left= 0.9cm of Alice3] {$\aProbabilisticDatabase(\aGraph) = \frac{1}{1024}$};
    \node [persona] (Bob3) [below right = of Alice3] {Bob};
    \node [persona] (Carl3) [below left = of Bob3] {Carl};
    \node [persona] (Dave3) [below left = of Alice3] {Dave};
    \node (moreGraphs) [right = 1cm of Bob3] {\Huge{. . .}};
    
    \node [draw = none] (Path_start3) [right=0.2cm of Bob3] {};
    \node [draw = none] (Path_start_up3) [above =1cm of Path_start3] {};
    
     \draw [<->, postaction= {decorate,decoration={text along path,text align=center,text={|\scriptsize|friend},raise=5pt}}] (Dave3) to (Alice3) node {};

   \draw [->, postaction= {decorate,decoration={text along path,text align=center,text={|\scriptsize|friend},raise=5pt, reverse path}}] (Carl3) to (Dave3) node {};
    \draw [->, postaction= {decorate,decoration={text along path,text align=center,text={|\scriptsize|friend},raise=5pt, reverse path}}] (Bob3) to (Carl3) node {}; 
  \end{scope}

  \begin{scope}[xshift=-4cm, yshift=-5.5cm, scale=0.6, every node/.append style={transform shape}]
    \node [persona] (Alice4) {Alice};
    \node [draw = none] (graph_clean_name4) [left= 0.9cm of Alice4] {$\aProbabilisticDatabase(\aGraph) = 0$};
    \node [persona] (Bob4) [below right = of Alice4] {Bob};
    \node [persona] (Carl4) [below left = of Bob4] {Carl};
    \node [persona] (Dave4) [below left = of Alice4] {Dave};
    
    \node [draw = none] (Path_start4) [right=0.2cm of Bob4] {};
    \node [draw = none] (Path_start_up4) [above =1cm of Path_start4] {};
    
     \draw [->, postaction= {decorate,decoration={text along path,text align=center,text={|\scriptsize|friend},raise=5pt}}] (Dave4) to (Alice4) node {};

   \draw [->, postaction= {decorate,decoration={text along path,text align=center,text={|\scriptsize|friend},raise=5pt, reverse path}}] (Bob4) to (Alice4) node {};
    \draw [->, postaction= {decorate,decoration={text along path,text align=center,text={|\scriptsize|friend},raise=5pt, reverse path}}] (Bob4) to (Carl4) node {}; 
  \end{scope}

  \begin{scope}[xshift=0cm, yshift=-5.5cm, scale=0.6, every node/.append style={transform shape}]
    \node [persona] (Alice5) {Alice};
    \node [draw = none] (graph_clean_name5) [left= 0.9cm of Alice5] {$\aProbabilisticDatabase(\aGraph) = 0$};
    \node [persona] (Bob5) [below right = of Alice5] {Bob};
    \node [persona] (Carl5) [below left = of Bob5] {Carl};
    \node [persona] (Dave5) [below left = of Alice5] {Dave};
    
    \node [draw = none] (Path_start5) [right=0.2cm of Bob5] {};
    \node [draw = none] (Path_start_up5) [above =1cm of Path_start5] {};
    
     \draw [->, postaction= {decorate,decoration={text along path,text align=center,text={|\scriptsize|friend},raise=5pt, reverse path}}] (Alice5) to (Dave5) node {}; 
    
   \draw [->, postaction= {decorate,decoration={text along path,text align=center,text={|\scriptsize|friend},raise=5pt}}] (Alice5) to (Bob5) node {};
   \draw [->, postaction= {decorate,decoration={text along path,text align=center,text={|\scriptsize|friend},raise=5pt, reverse path}}] (Alice5) to (Carl5) node {};
   \draw [->, postaction= {decorate,decoration={text along path,text align=center,text={|\scriptsize|friend},raise=5pt, reverse path}}] (Carl5) to (Dave5) node {};
    \draw [<->, postaction= {decorate,decoration={text along path,text align=center,text={|\scriptsize|friend},raise=5pt, reverse path}}] (Bob5) to (Carl5) node {}; 
  \end{scope}
  
  \begin{scope}[xshift=4cm, yshift=-5.5cm, scale=0.6, every node/.append style={transform shape}]
    \node [persona] (Alice6) {Alice};
    \node [draw = none] (graph_clean_name6) [left= 0.9cm of Alice6] {$\aProbabilisticDatabase(\aGraph) = 0$};
    \node [persona] (Bob6) [below right = of Alice6] {Bob};
    \node [persona] (Carl6) [below left = of Bob6] {Carl};
    \node [persona] (Dave6) [below left = of Alice6] {Dave};
    \node (moreGraphs6) [right = 1cm of Bob6] {\Huge{. . .}};
    
    \node [draw = none] (Path_start6) [right=0.2cm of Bob6] {};
    \node [draw = none] (Path_start_up6) [above =1cm of Path_start6] {};

   \draw [->, postaction= {decorate,decoration={text along path,text align=center,text={|\scriptsize|friend},raise=5pt, reverse path}}] (Carl6) to (Dave6) node {};
    \draw [->, postaction= {decorate,decoration={text along path,text align=center,text={|\scriptsize|friend},raise=5pt, reverse path}}] (Bob6) to (Carl6) node {}; 
   \draw [->, postaction= {decorate,decoration={text along path,text align=center,text={|\scriptsize|friend},raise=5pt, reverse path}}] (Alice6) to (Carl6) node {};    
  \end{scope}

  \begin{pgfonlayer}{background}
    \filldraw [line width=4mm,join=round,blue!20]
      (Alice.north  -| Bob.east)  rectangle (Carl.south  -| Dave.west);
    \filldraw [line width=4mm,join=round,blue!20]
      (Alice2.north  -| Bob2.east)  rectangle (Carl2.south  -| Dave2.west);
    \filldraw [line width=4mm,join=round,blue!20]
      (Alice3.north  -| Bob3.east)  rectangle (Carl3.south  -| Dave3.west); 
      
    \filldraw [line width=4mm,join=round,red!10]
      (Alice4.north  -| Bob4.east)  rectangle (Carl4.south  -| Dave4.west);     
    \filldraw [line width=4mm,join=round,red!10]
      (Alice5.north  -| Bob5.east)  rectangle (Carl5.south  -| Dave5.west);     
    \filldraw [line width=4mm,join=round,red!10]
      (Alice6.north  -| Bob6.east)  rectangle (Carl6.south  -| Dave6.west);     
  \end{pgfonlayer}  
\end{tikzpicture}
 \caption{A toy example of a probabilistic data-graph, where nodes $1,2,3,4$ are labeled with the $\Sigma_n$ data values $a_1=\esLabel{Alice}, a_2=\esLabel{Bob}, a_3=\esLabel{Carl}$ and $a_4=\esLabel{Dave}$, respectively.}  
 \label{Figure:AProbabilisticDataGraphToyExample}
\end{figure}
\end{example}

Having defined a probabilistic data-graph, we now present the notion of transformation among possible states of a data-graph, 
via a probabilistic criteria.

\begin{definition}\label{def:realizationAndCosuport}
Given by $\Sigma_e$ and $\Sigma_n$, let $\aUniverseOfGraphs$ be the universe of data-graphs and $\ensuremath{P}$ be the universe of probabilistic data-graphs. 
A \defstyle{realization model} (or a \defstyle{noisy observer}) is a function $\aRealizationModel: \aUniverseOfGraphs \to \ensuremath{P}$. That is, $\aRealizationModel$ maps data-graphs $\aGraph\in\aUniverseOfGraphs$ into probabilistic data-graphs.
\end{definition}

Intuitively, $\aRealizationModel$
allows us to ‘observe’ a potentially distorted version of a clean data-graph.
We now formalize the Epistemic Model of a Graph Database.

\begin{definition}[EMDG] \label{def:EMDG}
Let $\aUniverseOfGraphs$ be the universe of data-graphs given by $\Sigma_e$ and $\Sigma_n$. We define an \defstyle{Epistemic Model of a Data-graph (EMDG)} as a pair $(\aProbabilisticDatabase,\aRealizationModel)$, where $\aProbabilisticDatabase$ is a probabilistic data-graph over $\aUniverseOfGraphs$, and $\aRealizationModel$ is a noisy observer. 
\end{definition}

Given $\aGraph'$ and $\aRealizationModel$, we define the \defstyle{cosupport of $G'$} as the set $\{\aGraph \in  U  \mid  \aGraph' \in \aRealizationModel(\aGraph) \}$. We denote its cardinality by $\sigma_{\aRealizationModel}(\aGraph')$, or just $\sigma(\aGraph')$ if $\aRealizationModel$ is clear from the context.
Throughout this paper, we will consider always EMDGs having finite cosupport for every $\aGraph'$.

We denote by $\aInverseRealizationModel_{\aProbabilisticDatabase, \aRealizationModel}$ the function whose domain is $\aUniverseOfGraphs$ and whose image for each data-graph $\aGraph'$ is a particular probabilistic data-graph, such that  $\aInverseRealizationModel_{\aProbabilisticDatabase, \aRealizationModel}(\aGraph')( \aGraph)$ is the probability that the data-graph $\aGraph$ was the clean data-graph transformed via the noisy observer $\aRealizationModel$ into the observed data-graph $\aGraph'$ taking into account the probabilities defined by $\aProbabilisticDatabase$. More precisely, and using the intuition of Bayes' theorem, we define 

\begin{equation}\label{eq:Bayes}
    \aInverseRealizationModel_{\aProbabilisticDatabase,\aRealizationModel}(\aGraph')(\aGraph)  \eqdef \frac{\aProbabilisticDatabase(\aGraph) \times \aRealizationModel(\aGraph)(\aGraph') }{\sum_{\aGraphb \in \aProbabilisticDatabase} (\aProbabilisticDatabase(\aGraphb)\times \aRealizationModel(\aGraphb)(\aGraph'))}
\end{equation}
whenever there exists $\aGraphb \in \aProbabilisticDatabase$ such that $\aGraph'\in \aRealizationModel(\aGraphb)$, 
and 0 otherwise. Notice that in the former case, (\ref{eq:Bayes}) is well-defined as a probabilistic data-graph given that the cosupport of $\aGraph'$ is finite.
Here, given the observed data-graph
$\aGraph'$ and a data-graph $\aGraph$,
$\aRealizationModel(\aGraph)(\aGraph')$
is the probability that $\aGraph'$ is the graph
that results from transforming $\aGraph$ 
through $\aRealizationModel$. Note that this does not take into account the prior given by 
$\aProbabilisticDatabase(\aGraph)$. The denominator is the normalization factor over all data-graphs with positive probability in
$\aProbabilisticDatabase$ that could lead to 
$\aGraph'$ through $\aRealizationModel$.



\begin{remark}
Intuitively, $\aInverseRealizationModel_{\aProbabilisticDatabase, \aRealizationModel}(\aGraph')(\aGraph)$ is the conditional probability of $\aGraph$ given the observation of $\aGraph'$, taking into account the distribution over possible states $\aProbabilisticDatabase$ and the realization model $\aRealizationModel$. We may write $\aInverseRealizationModel$ without the subscript when both $\aProbabilisticDatabase$ and $\aRealizationModel$ are clear from the context. Even though we defined $\aInverseRealizationModel(\aGraph')$ as 0 when the cosupport of $\aGraph'$ intersected with the support of $\aProbabilisticDatabase$ is empty, in the proofs and examples we mostly consider and argue about the (much more interesting) non-empty case, so it will be common to refer to $\aInverseRealizationModel(\aGraph')$ as a probabilistic data-graph regardless of the observed data-graph $\aGraph'$. 
\end{remark}

As it is usually the case, we are going to assume that $\aProbabilisticDatabase$, $\aRealizationModel$, and $\aInverseRealizationModel$ are all computable functions over adequate representations of their domains and codomains. 
When proving complexity bounds for problems related to EMDGs (or rather PUDGs, defined later in this work) we will be concerned with the set of probabilistic databases and realization models that can be computed efficiently. This means that, given $\aGraph$, we consider $\aProbabilisticDatabase$'s and $\aRealizationModel$'s such that we can compute $\aProbabilisticDatabase(\aGraph)$ and $\aRealizationModel(\aGraph)$ in time $poly(|\aGraph|)$. We also ask for the distribution $\aInverseRealizationModel(\aGraph')$ to be computable in polynomial time given its input and output: that is, we can compute it in $poly(|\aGraph'| +|\aInverseRealizationModel(\aGraph')|)$. Note that defining an arbitrary efficiently computable $\aProbabilisticDatabase$ and $\aRealizationModel$ does not imply that $\aInverseRealizationModel(\aGraph')(\aGraph)$ can be computed in polynomial time.  

\begin{example}\label{example:toySocialNetworkAndNoisyObserverOverGraph}
A particular example of the previous possible interpretation of an EMDG is shown in Figure~\ref{figure:exampleARealizationModel}, where a toy social network is incompletely observed via the `deletion' of a single edge. Here, the data values (in $\Sigma_n$) are name identifiers $\esDato{Alice}, \esDato{Bob}, \esDato{Carl}$, and the possible edge-labels (in $\Sigma_e$) are $\esLabel{friend}$ and $\esLabel{follows}$. In this case, one possible action of $\aRealizationModel$ over the data-graph $\aGraph$ is leaving it unchanged, while the other shown possibility is the deletion of a $\esLabel{friend}$ edge in the direction from Bob to Alice. Assuming that the representation in the figure is complete, we have that $\aRealizationModel(\aGraph)(\aGraph) + \aRealizationModel(\aGraph)(\aGraph') = 1$, since $\aRealizationModel(\aGraph)$ is a probabilistic data-graph.
\begin{figure}[!ht]%
\centering 
\begin{tikzpicture}[node distance=3.5cm,->,>=stealth', bend angle=45,auto] 
  \tikzstyle{persona}=[rounded rectangle,thick,draw=black!80,fill=black!0,minimum size=10mm]
  \tikzstyle{every label}=[blue]
  \begin{scope}
    \node [persona] (Alice) {Alice};
    \node [draw = none] (graph_clean_name) [right= 1.5cm of Alice] {$\aGraph \in \aProbabilisticDatabase$};
    \node [persona] (Bob) [below right = of Alice] {Bob};
    \node [persona] (Carl) [below left = of Bob] {Carl};
    
    \node [draw = none] (Path_start) [right=0.2cm of Bob] {};
    \node [draw = none] (Path_start_up) [above =1cm of Path_start] {};
    
    \draw [->, postaction= {decorate,decoration={text along path,text align=center,text={friend},raise=5pt}}, bend right] (Alice) to (Bob) node {}; 
    \draw [->, postaction= {decorate,decoration={text along path,text align=center,text={friend},raise=5pt, reverse path}}, bend right] (Bob) to (Alice) node {};  
   \draw [->, postaction= {decorate,decoration={text along path,text align=center,text={follows},raise=5pt, reverse path}}] (Bob) to (Alice) node {};
      \draw [->, postaction= {decorate,decoration={text along path,text align=center,text={follows},raise=5pt, reverse path}}] (Bob) to (Carl) node {}; 
  \end{scope}

\tikzstyle{persona}=[rounded rectangle,thick,draw=black!80,fill=black!5,minimum size=10mm] 
  \begin{scope}[xshift=6cm, yshift=1.5cm, scale=0.6, every node/.append style={transform shape}]
    \node [persona] (Alice2) {Alice};
    \node [draw = none] (graph_unclean_name) [right= 2cm of Alice2] {$\aGraph$};
    \node [persona] (Bob2) [below right = of Alice2] {Bob};
    \node [persona] (Carl2) [below left = of Bob2] {Carl};
    
    \node [draw = none] (Path_end) [left=2.5cm of Bob2] {};
    
    \draw [->, postaction= {decorate,decoration={text along path,text align=center,text={|\scriptsize|friend},raise=5pt}}, bend right] (Alice2) to (Bob2) node {}; 
   \draw [->, postaction= {decorate,decoration={text along path,text align=center,text={|\scriptsize|follows},raise=5pt, reverse path}}] (Bob2) to (Alice2) node {};
    \draw [->, postaction= {decorate,decoration={text along path,text align=center,text={|\scriptsize|friend},raise=5pt, reverse path}}, bend right] (Bob2) to (Alice2) node {};    
   \draw [->, postaction= {decorate,decoration={text along path,text align=center,text={|\scriptsize|follows},raise=5pt, reverse path}}] (Bob2) to (Carl2) node {};
  \end{scope}

   \begin{scope}[xshift=6cm, yshift=-4cm, scale=0.6, every node/.append style={transform shape}]
    \node [persona] (Alice3) {Alice};
    \node [draw = none] (graph_unclean_name) [right= 2cm of Alice3] {$\aGraph'$};
    \node [persona] (Bob3) [below right = of Alice3] {Bob};
    \node [persona] (Carl3) [below left = of Bob3] {Carl};
    
    \node [draw = none] (Path_end2) [left=2.5cm of Bob3] {};
    
    \draw [->, postaction= {decorate,decoration={text along path,text align=center,text={|\scriptsize|friend},raise=5pt}}, bend right] (Alice3) to (Bob3) node {}; 
   \draw [->, postaction= {decorate,decoration={text along path,text align=center,text={|\scriptsize|follows},raise=5pt, reverse path}}] (Bob3) to (Alice3) node {};
   \draw [->, postaction= {decorate,decoration={text along path,text align=center,text={|\scriptsize|follows},raise=5pt, reverse path}}] (Bob3) to (Carl3) node {};
  \end{scope}

  \draw [thick,decorate, decoration=snake, segment length=10mm]
    (Path_start) ->  (Path_end) 
    node [above=3mm,midway,text width=3cm,text centered]
      {$\aRealizationModel$};
  \draw [thick,decorate, decoration=snake, segment length=10mm]
    (Path_start) ->  (Path_end2) 
    node [above=3mm,midway,text width=3cm,text centered]
      {$\aRealizationModel$};      

  \begin{pgfonlayer}{background}
    \filldraw [line width=4mm,join=round,blue!20]
      (Alice.north  -| Bob.east)  rectangle (Carl.south  -| Carl.west);
    \filldraw [line width=4mm,join=round,blue!10]      
      (Alice2.north -| Bob2.east) rectangle (Carl2.south -| Carl2.west)
      (Alice3.north -| Bob3.east) rectangle (Carl3.south -| Carl3.west);      
  \end{pgfonlayer}  
\end{tikzpicture}
\caption{Two possible actions of a particular noisy observer $\aRealizationModel$ over a given data-graph $\aGraph$ in $\aProbabilisticDatabase$. In one case, with probability $\aRealizationModel(\aGraph)(\aGraph)$, the original database is left unchanged. In the other case, with probability $\aRealizationModel(\aGraph)(\aGraph')$, one particular edge of the original database is deleted.}  
\label{figure:exampleARealizationModel}
\end{figure}
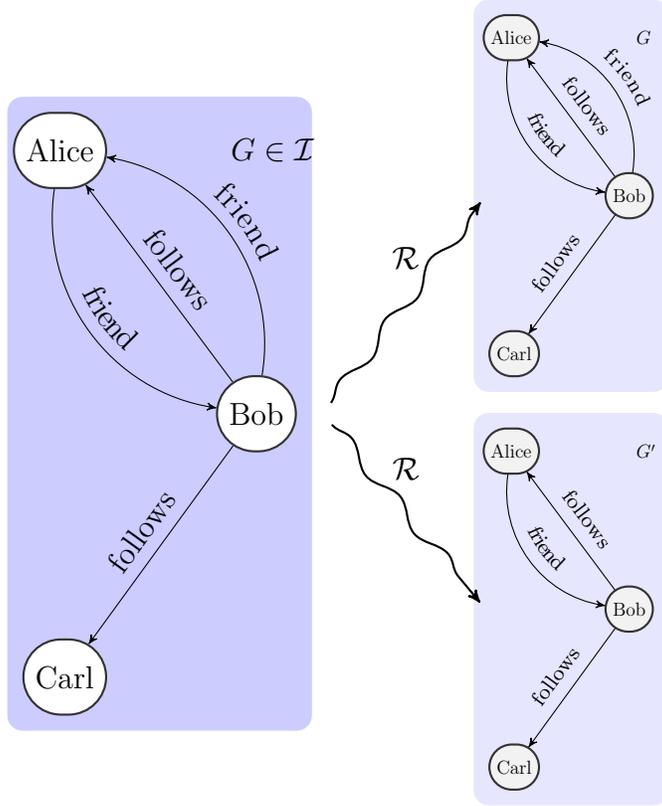
\end{example}

Another possible interpretation for EMDGs is that $\aProbabilisticDatabase$ represents idealized databases over certain domain, and $\aRealizationModel$ transforms each idealized version into many possible error-prone implementations. For example, $\aProbabilisticDatabase$ could represent bibliographical databases, and $\aRealizationModel$ could produce databases where some type of data-duplication errors have been made. 

Inspired by~\cite{rekatsinas2017holoclean, de2018formal}, we define a probabilistic unclean data-graph in the following way:
\begin{definition} [PUDG]\label{def:PUDG}
We define a \defstyle{Probabilistic Unclean Data-graph (PUDG)} as a triple $\aPUDG = (\aProbabilisticDatabase,\aRealizationModel,\aGraph')$, where: 
    \begin{itemize}
    \item ($\aProbabilisticDatabase,\aRealizationModel$) is an EMDG.
    \item $\aGraph'$ is the observed or unclean data-graph.
    \end{itemize}
\end{definition}

Intuitively, $\aGraph'$ is the current observable state of the database, which we suspect incorrect because of errors introduced due to data entry mistakes, data integration problems, or other kind of operations that could have polluted the original (unknown) data-graph~$\aGraph$. We have some knowledge about the original state of the data-graph prior to these issues, which is modeled by the distribution $\aProbabilisticDatabase$; and we also have some idea of the possible ways these original  data-graphs from $\aProbabilisticDatabase$ could have become unclean, which is modeled by $\aRealizationModel$.

\begin{example}\label{example:SocialNetwork}
Let $\Sigma_e = \{ \esLabel{friend\_of} \}$, $\Sigma_n$ be a set of common names. 
Consider $S_N$ the set of finite data-graphs over $\Sigma_e, \Sigma_n$ that have at least 2 nodes and no more than $N$ nodes. 
Let $\aProbabilisticDatabase : \aUniverseOfGraphs \to [0,1]$ be a uniform distribution of probabilities among all graphs of $S_N$, having $\aProbabilisticDatabase(\aGraph) = 0$ if $\aGraph \in \aUniverseOfGraphs \backslash S_N$, and otherwise with $\aProbabilisticDatabase(\aGraph)$ given by some prior over the typical shape of social networks.

We can use $\aRealizationModel$ to model the probability of erroneously removing a friendship relation between two people. For example, it could be the case that the probability of this error is a fixed number $p$ between $0$ and $1$, and that the probability of making this error is independent for each edge; therefore, larger databases would usually be observed with more errors. We can represent this idea by defining $\aRealizationModel$ such that, for any unclean observed state $\aGraph'$ and data-graph $\aGraph$, we have:
$$
\aRealizationModel(\aGraph) (\aGraph') = \begin{cases} p^{|E_\aGraph \setminus E_{\aGraph'}|} \times (1-p)^{|E_{\aGraph'}|}& \text{if $\aGraph' \subseteq \aGraph$ and $V_\aGraph = V_\aGraph'$}\\
0 & \text{otherwise}
\end{cases}
$$

Note that $\aRealizationModel$, as defined above, is indeed a realization model. That is, for every $\aGraph$, $\aRealizationModel(\aGraph)$ is a probabilistic data graph: 
$$\sum_{\aGraph' \in \aUniverseOfGraphs} \aRealizationModel(\aGraph)(\aGraph')= \sum_{\aGraph' \subseteq \aGraph, V_\aGraph = V_{\aGraph'}} p^{|E_\aGraph \setminus E_{\aGraph'}|} \times (1-p)^{|E_{\aGraph'}|} = 1$$

Furthermore, note that this way of constructing $\aRealizationModel$ can be generalized naturally to larger sets of edge labels $\Sigma_e$ by assigning different values $p_\aLabel$ to each $\aLabel \in \Sigma_e$.

\end{example}

\begin{example}\label{example:BibliographicalRIntroducesAuthorship}
Let $\Sigma_e = \{\esLabel{author\_of}, \esLabel{name}, \esLabel{title}\}$, and let $\Sigma_n$ be the set of strings over the English alphabet. 
We can consider an 
EMDG $(\aProbabilisticDatabase,\aRealizationModel)$, where $\aProbabilisticDatabase$ assigns non-zero probability to all data-graphs that form a viable representation of bibliographical data (captured using some path or tree expression), and where $\aRealizationModel$ can introduce errors by adding outgoing $\esLabel{author\_of}$ edges from any node that has no incoming edges. For example, $\aRealizationModel$ could erroneously assign co-authorship of the same book, or it could (in a more blatant error) introduce authorship between a person and a name, or between a person and another person.

If we have a PUDG $\aPUDG = (\aProbabilisticDatabase,\aRealizationModel,\aGraph')$, then the introduction of errors might cause the observed $\aGraph'$ not to coincide with the initial clean data-graph. Nonetheless, the nature of this $\aRealizationModel$ allows us to know that the true answer must be a subgraph of $\aGraph'$ where the only possible changes, if any, are the removal of $\esLabel{author\_of}$ edges. Furthermore, the information on $\aProbabilisticDatabase$ can also be used; for example, by its own nature we know with certainty that any $\esLabel{author\_of}$ edges between different authors are errors introduced by $\aRealizationModel$.
\end{example}

In principle, $\aRealizationModel$ might represent any possible transformation between data-graphs. In order to restrict the generality of $\aRealizationModel$ while still preserving a good amount of expressiveness so it remains useful for real-world applications, we will consider realization models that either only delete, only add, or only modify data in the data-graphs (either nodes or edges or data). Considering any of these restrictions is an usual practice when reasoning over uncertainty or inconsistencies and, furthermore, 
all three notions can be related to common semantics for repairing databases~\cite{afrati2009repair,barcelo2017data, chomicki2005minimal, Cate:2015}. 

In order to formalize these ideas, we start by defining different types of EMDG that characterize the different realization models we seek to capture.

\begin{definition}
A \defstyle{Subset EMDG} is a pair $(\aProbabilisticDatabase, \aRealizationModel)$ such that the noisy observer $\aRealizationModel$ satisfies that $\aGraph' \subseteq \aGraph$ for every data-graph $\aGraph'$ such that $\aRealizationModel(\aGraph)(\aGraph') > 0$.

\end{definition}

In a way, we can think that the noisy observer of a Subset EMDG introduces node or edge deletions.
To formalize the second notion we define {Superset EMDG} in an analogous way:

\begin{definition}
A \defstyle{Superset EMDG} is a pair $(\aProbabilisticDatabase, \aRealizationModel)$ such that the noisy observer $\aRealizationModel$ satisfies that $\aGraph' \supseteq \aGraph$ for every
data-graph $\aGraph'$ such that $\aRealizationModel(\aGraph)(\aGraph') > 0$.
 \end{definition}

In other words, in this case the noisy observer can add nodes or edges to the original clean graph.



For the case where data modifications are allowed, we provide now formal definitions for two different kinds of  data modifications.

\begin{definition}\label{def:update} 
We say that $(\aProbabilisticDatabase, \aRealizationModel)$ is a \defstyle{Node-Update EMDG} when $\aRealizationModel$ is such that if $\aRealizationModel(\aGraph)(\aGraph') > 0$, then $V_\aGraph = V_{\aGraph'}$ and $L_\aGraph(v,w) = L_{\aGraph'}(v,w)$ for all $u,v\in V_\aGraph$.

On the other hand, we say that $(\aProbabilisticDatabase, \aRealizationModel)$ is an \defstyle{Update EMDG} when $\aRealizationModel$ is such that if $\aRealizationModel(\aGraph)(\aGraph') > 0$, then $V_\aGraph = V_{\aGraph'}$ and $ |L_\aGraph(v,w)| = |L_{\aGraph'}(v,w)|$ for all $u,v\in V_\aGraph$.

\end{definition}

Intuitively, a node-update only modifies the data-values in the original data-graph, leaving all other structure intact, that is, only the function $D$ might change. Figure~\ref{figure:exampleNodeUpdate} shows a toy example of a social network where the names of the individuals can be altered by the noisy observer $\aRealizationModel$, while leaving the rest of the network's structure intact. 
Compare with Figure~\ref{figure:exampleARealizationModel}, and more generally with Subset and Superset EMDGs, where the observer is able to modify the data-graph via the removal or addition of edges, but without changing the data in the nodes. 
Analogously, we may understand an Update EMDG as a transformation that is allowed to modify both the data-values from the nodes and the labels from the edges, while leaving intact the remaining structure (including the number of edges). Figure~\ref{figure:exampleNodeEdgeUpdate} presents a simple example where more changes are possible than in Figure~\ref{figure:exampleNodeUpdate}.

\begin{definition}[$\Pi$ PUDG]
A \defstyle{$\Pi$ PUDG} is a triple $\aPUDG =(\aProbabilisticDatabase,\aRealizationModel,\aGraph')$ where $(\aProbabilisticDatabase, \aRealizationModel)$ is a $\Pi$ EMDG where $\Pi \in \{\text{Subset}, \text{Superset}, \text{Update}, \text{Node-Update}\}$.
\end{definition}


\begin{example} 
An Update PUDG can be used to model uncertainty that is focused on particular portions of the data-graph. For example, $\Sigma_n, \Sigma_e$ can contain distinguished symbols (such as \esDato{Name?} and \esLabel{?}) which denote unknown data. Our prior \aProbabilisticDatabase might represent a distribution of probabilities over the real state of the world, with no data-graph in $\aProbabilisticDatabase$ containing one of the uncertainty labels. Then, we could use a noisy observer $\aRealizationModel$ that only modifies a graph by transforming precise data values or labels into unknown ones. 
See Figure~\ref{figure:exampleEpistemicSymbolsNodeEdgeUpdate} for a toy representation of this kind of modelling.

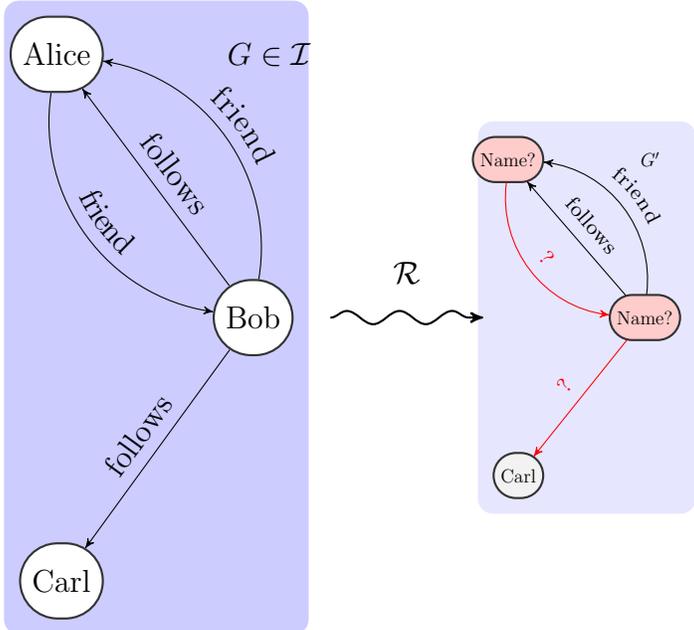
\begin{figure}[ht!]%
\centering 
\begin{tikzpicture}[node distance=3.5cm,->,>=stealth', bend angle=45,auto] 
  \tikzstyle{persona}=[rounded rectangle,thick,draw=black!80,fill=black!0,minimum size=10mm]
  \tikzstyle{every label}=[blue]
  \begin{scope}
    \node [persona] (Alice) {Alice};
    \node [draw = none] (graph_clean_name) [right= 1.5cm of Alice] {$\aGraph \in \aProbabilisticDatabase$};
    \node [persona] (Bob) [below right = of Alice] {Bob};
    \node [persona] (Carl) [below left = of Bob] {Carl};
    
    \node [draw = none] (Path_start) [right=0.2cm of Bob] {};
    \node [draw = none] (Path_start_up) [above =1cm of Path_start] {};
    
    \draw [->, postaction= {decorate,decoration={text along path,text align=center,text={friend},raise=5pt}}, bend right] (Alice) to (Bob) node {}; 
    \draw [->, postaction= {decorate,decoration={text along path,text align=center,text={friend},raise=5pt, reverse path}}, bend right] (Bob) to (Alice) node {};  
   \draw [->, postaction= {decorate,decoration={text along path,text align=center,text={follows},raise=5pt, reverse path}}] (Bob) to (Alice) node {};
      \draw [->, postaction= {decorate,decoration={text along path,text align=center,text={follows},raise=5pt, reverse path}}] (Bob) to (Carl) node {}; 
  \end{scope}

\tikzstyle{persona}=[rounded rectangle,thick,draw=black!80,fill=black!5,minimum size=10mm] 
    
    

   \begin{scope}[xshift=6cm, yshift=-1.4cm, scale=0.6, every node/.append style={transform shape}]
    \node [persona, fill=red!20] (Alice3) {Name?};
    \node [draw = none] (graph_unclean_name) [right= 2cm of Alice3] {$\aGraph'$};
    \node [persona, fill=red!20] (Bob3) [below right = of Alice3] {Name?};
    \node [persona] (Carl3) [below left = of Bob3] {Carl};
    
    \node [draw = none] (Path_end2) [left=2.5cm of Bob3] {};
    
    \draw [->, draw = red, postaction= {decorate,decoration={text along path,text align=center, text={|\scriptsize\color{red}|?},raise=5pt}}, bend right] (Alice3) to (Bob3) node {};  
   \draw [->, postaction= {decorate,decoration={text along path,text align=center,text={|\scriptsize|follows},raise=5pt, reverse path}}] (Bob3) to (Alice3) node {};
    \draw [->, postaction= {decorate,decoration={text along path,text align=center,text={|\scriptsize|friend},raise=5pt, reverse path}}, bend right] (Bob3) to (Alice3) node {};      
   \draw [->, draw = red, postaction= {decorate,decoration={text along path,text align=center,text={|\scriptsize\color{red}|?},raise=5pt, reverse path}}] (Bob3) to (Carl3) node {};
  \end{scope}

  \draw [thick,decorate, decoration=snake, segment length=7mm] 
    (Path_start) ->  (Path_end2) 
    node [above=3mm,midway,text width=3cm,text centered]
      {$\aRealizationModel$};      

  \begin{pgfonlayer}{background}
    \filldraw [line width=4mm,join=round,blue!20]
      (Alice.north  -| Bob.east)  rectangle (Carl.south  -| Carl.west);
    \filldraw [line width=4mm,join=round,blue!10]      
      (Alice3.north -| Bob3.east) rectangle (Carl3.south -| Carl3.west);      
  \end{pgfonlayer}  
\end{tikzpicture}
 \caption{An Update PUGD with distinguished symbols to pinpoint which data is uncertain.}  
 \label{figure:exampleEpistemicSymbolsNodeEdgeUpdate}
\end{figure}
\end{example}

\newcommand{\ScaleFactorForExampleTwoPUGDs}{0.7} 
\newcommand{\SizeMiniPagesForExampleTwoPUGDs}{.48\textwidth} 

\begin{figure}[!htb]
    \centering
    \begin{subfigure}{\SizeMiniPagesForExampleTwoPUGDs}
        \centering
        \scalebox{\ScaleFactorForExampleTwoPUGDs}
{
\begin{tikzpicture}[node distance=3.5cm,->,>=stealth', bend angle=45,auto] 
  \tikzstyle{persona}=[rounded rectangle,thick,draw=black!80,fill=black!0,minimum size=10mm]
  \tikzstyle{every label}=[blue]
  \begin{scope}
    \node [persona] (Alice) {Alice};
    \node [draw = none] (graph_clean_name) [right= 1.5cm of Alice] {$\aGraph \in \aProbabilisticDatabase$};
    \node [persona] (Bob) [below right = of Alice] {Bob};
    \node [persona] (Carl) [below left = of Bob] {Carl};
    
    \node [draw = none] (Path_start) [right=0.2cm of Bob] {};
    \node [draw = none] (Path_start_up) [above =1cm of Path_start] {};
    
    \draw [->, postaction= {decorate,decoration={text along path,text align=center,text={friend},raise=5pt}}, bend right] (Alice) to (Bob) node {}; 
    \draw [->, postaction= {decorate,decoration={text along path,text align=center,text={friend},raise=5pt, reverse path}}, bend right] (Bob) to (Alice) node {};  
   \draw [->, postaction= {decorate,decoration={text along path,text align=center,text={follows},raise=5pt, reverse path}}] (Bob) to (Alice) node {};
      \draw [->, postaction= {decorate,decoration={text along path,text align=center,text={follows},raise=5pt, reverse path}}] (Bob) to (Carl) node {}; 
  \end{scope}

\tikzstyle{persona}=[rounded rectangle,thick,draw=black!80,fill=black!5,minimum size=10mm] 
  \begin{scope}[xshift=6cm, yshift=1.5cm, scale=0.6, every node/.append style={transform shape}]
    \node [persona] (Alice2) {Alice};
    \node [draw = none] (graph_unclean_name) [right= 2cm of Alice2] {$\aGraph$};
    \node [persona] (Bob2) [below right = of Alice2] {Bob};
    \node [persona] (Carl2) [below left = of Bob2] {Carl};
    
    \node [draw = none] (Path_end) [left=2.5cm of Bob2] {};
    
    \draw [->, postaction= {decorate,decoration={text along path,text align=center,text={|\scriptsize|friend},raise=5pt}}, bend right] (Alice2) to (Bob2) node {}; 
   \draw [->, postaction= {decorate,decoration={text along path,text align=center,text={|\scriptsize|follows},raise=5pt, reverse path}}] (Bob2) to (Alice2) node {};
    \draw [->, postaction= {decorate,decoration={text along path,text align=center,text={|\scriptsize|friend},raise=5pt, reverse path}}, bend right] (Bob2) to (Alice2) node {};    
   \draw [->, postaction= {decorate,decoration={text along path,text align=center,text={|\scriptsize|follows},raise=5pt, reverse path}}] (Bob2) to (Carl2) node {};
  \end{scope}

   \begin{scope}[xshift=6cm, yshift=-4cm, scale=0.6, every node/.append style={transform shape}]
    \node [persona, fill=red!20] (Alice3) {Alise};
    \node [draw = none] (graph_unclean_name) [right= 2cm of Alice3] {$\aGraph'$};
    \node [persona, fill=red!20] (Bob3) [below right = of Alice3] {Carl};
    \node [persona] (Carl3) [below left = of Bob3] {Carl};
    
    \node [draw = none] (Path_end2) [left=2.5cm of Bob3] {};
    
    \draw [->, postaction= {decorate,decoration={text along path,text align=center,text={|\scriptsize|friend},raise=5pt}}, bend right] (Alice3) to (Bob3) node {}; 
   \draw [->, postaction= {decorate,decoration={text along path,text align=center,text={|\scriptsize|follows},raise=5pt, reverse path}}] (Bob3) to (Alice3) node {};
    \draw [->, postaction= {decorate,decoration={text along path,text align=center,text={|\scriptsize|friend},raise=5pt, reverse path}}, bend right] (Bob3) to (Alice3) node {};      
   \draw [->, postaction= {decorate,decoration={text along path,text align=center,text={|\scriptsize|follows},raise=5pt, reverse path}}] (Bob3) to (Carl3) node {};
  \end{scope}

  \draw [thick,decorate, decoration=snake, segment length=10mm]
    (Path_start) ->  (Path_end) 
    node [above=3mm,midway,text width=3cm,text centered]
      {$\aRealizationModel$};
  \draw [thick,decorate, decoration=snake, segment length=10mm]
    (Path_start) ->  (Path_end2) 
    node [above=3mm,midway,text width=3cm,text centered]
      {$\aRealizationModel$};      

  \begin{pgfonlayer}{background}
    \filldraw [line width=4mm,join=round,blue!20]
      (Alice.north  -| Bob.east)  rectangle (Carl.south  -| Carl.west);
    \filldraw [line width=4mm,join=round,blue!10]      
      (Alice2.north -| Bob2.east) rectangle (Carl2.south -| Carl2.west)
      (Alice3.north -| Bob3.east) rectangle (Carl3.south -| Carl3.west);      
  \end{pgfonlayer}  
\end{tikzpicture}
}
\caption{Here the $\aRealizationModel$ is a Node-Update noisy observer, which can introduce 
changes to the original data values in the nodes, but does not otherwise modify the original data-graph.}  
\label{figure:exampleNodeUpdate}
    \end{subfigure}
\hspace{0.15cm}\vline\hspace{0.15cm}
    \begin{subfigure}{\SizeMiniPagesForExampleTwoPUGDs}
        \centering
        \scalebox{\ScaleFactorForExampleTwoPUGDs}{

\begin{tikzpicture}[node distance=3.5cm,->,>=stealth', bend angle=45,auto] 
  \tikzstyle{persona}=[rounded rectangle,thick,draw=black!80,fill=black!0,minimum size=10mm]
  \tikzstyle{every label}=[blue]
  \begin{scope}
    \node [persona] (Alice) {Alice};
    \node [draw = none] (graph_clean_name) [right= 1.5cm of Alice] {$\aGraph \in \aProbabilisticDatabase$};
    \node [persona] (Bob) [below right = of Alice] {Bob};
    \node [persona] (Carl) [below left = of Bob] {Carl};
    
    \node [draw = none] (Path_start) [right=0.2cm of Bob] {};
    \node [draw = none] (Path_start_up) [above =1cm of Path_start] {};
    
    \draw [->, postaction= {decorate,decoration={text along path,text align=center,text={friend},raise=5pt}}, bend right] (Alice) to (Bob) node {}; 
    \draw [->, postaction= {decorate,decoration={text along path,text align=center,text={friend},raise=5pt, reverse path}}, bend right] (Bob) to (Alice) node {};  
   \draw [->, postaction= {decorate,decoration={text along path,text align=center,text={follows},raise=5pt, reverse path}}] (Bob) to (Alice) node {};
      \draw [->, postaction= {decorate,decoration={text along path,text align=center,text={follows},raise=5pt, reverse path}}] (Bob) to (Carl) node {}; 
  \end{scope}

\tikzstyle{persona}=[rounded rectangle,thick,draw=black!80,fill=black!5,minimum size=10mm] 
  \begin{scope}[xshift=6cm, yshift=1.5cm, scale=0.6, every node/.append style={transform shape}]
    \node [persona] (Alice2) {Alice};
    \node [draw = none] (graph_unclean_name) [right= 2cm of Alice2] {$\aGraph$};
    \node [persona] (Bob2) [below right = of Alice2] {Bob};
    \node [persona] (Carl2) [below left = of Bob2] {Carl};
    
    \node [draw = none] (Path_end) [left=2.5cm of Bob2] {};
    
    \draw [->, postaction= {decorate,decoration={text along path,text align=center,text={|\scriptsize|friend},raise=5pt}}, bend right] (Alice2) to (Bob2) node {}; 
   \draw [->, postaction= {decorate,decoration={text along path,text align=center,text={|\scriptsize|follows},raise=5pt, reverse path}}] (Bob2) to (Alice2) node {};
    \draw [->, postaction= {decorate,decoration={text along path,text align=center,text={|\scriptsize|friend},raise=5pt, reverse path}}, bend right] (Bob2) to (Alice2) node {};    
   \draw [->, draw = red, postaction= {decorate,decoration={text along path,text align=center,text={|\scriptsize\color{red}|friend},raise=5pt, reverse path}}] (Bob2) to (Carl2) node {};
  \end{scope}

   \begin{scope}[xshift=6cm, yshift=-4cm, scale=0.6, every node/.append style={transform shape}]
    \node [persona, fill=red!20] (Alice3) {Alise};
    \node [draw = none] (graph_unclean_name) [right= 2cm of Alice3] {$\aGraph'$};
    \node [persona, fill=red!20] (Bob3) [below right = of Alice3] {Carl};
    \node [persona] (Carl3) [below left = of Bob3] {Carl};
    
    \node [draw = none] (Path_end2) [left=2.5cm of Bob3] {};
    
    \draw [->, draw = red, postaction= {decorate,decoration={text along path,text align=center, text={|\scriptsize\color{red}|follows},raise=5pt}}, bend right] (Alice3) to (Bob3) node {};  
   \draw [->, postaction= {decorate,decoration={text along path,text align=center,text={|\scriptsize|follows},raise=5pt, reverse path}}] (Bob3) to (Alice3) node {};
    \draw [->, postaction= {decorate,decoration={text along path,text align=center,text={|\scriptsize|friend},raise=5pt, reverse path}}, bend right] (Bob3) to (Alice3) node {};      
   \draw [->, draw = red, postaction= {decorate,decoration={text along path,text align=center,text={|\scriptsize\color{red}|friend},raise=5pt, reverse path}}] (Bob3) to (Carl3) node {};
  \end{scope}

  \draw [thick,decorate, decoration=snake, segment length=10mm]
    (Path_start) ->  (Path_end) 
    node [above=3mm,midway,text width=3cm,text centered]
      {$\aRealizationModel$};
  \draw [thick,decorate, decoration=snake, segment length=10mm]
    (Path_start) ->  (Path_end2) 
    node [above=3mm,midway,text width=3cm,text centered]
      {$\aRealizationModel$};      

  \begin{pgfonlayer}{background}
    \filldraw [line width=4mm,join=round,blue!20]
      (Alice.north  -| Bob.east)  rectangle (Carl.south  -| Carl.west);
    \filldraw [line width=4mm,join=round,blue!10]      
      (Alice2.north -| Bob2.east) rectangle (Carl2.south -| Carl2.west)
      (Alice3.north -| Bob3.east) rectangle (Carl3.south -| Carl3.west);      
  \end{pgfonlayer}  
\end{tikzpicture}
}
\caption{Here the $\aRealizationModel$ is a general Update noisy observer. It can introduce changes to both the data in the nodes and to the edge labels, but it does not otherwise modify the original data-graph.}  
\label{figure:exampleNodeEdgeUpdate}
    \end{subfigure}
    \caption{Two types of Update EMDGs, with different restrictions on the function $\aRealizationModel$.}
\end{figure}

It is common to expect databases to follow some semantic structure related to the world they are representing. This structure can be enforced by defining a set $\aRestrictionSet$ of integrity constraints that limit the shape and data contained in the database. In the case of data-graphs, a common way of defining these constraints is through \textit{path constraints} defined in some specific logical language. Given a constraint $\aFormula$, we want our database to be consistent with respect to $\aFormula$, given some definition of consistency based on the semantics of $\aFormula$.

In this work, constraints will be defined with the \defstyle{\Gregxpath} language. These formulas can capture both nodes and pairs of nodes based upon expressions that
allow to define paths along the graph with common regular expression operators as well as being able to compare data values.

\textit{Path expressions} of \Gregxpath are given by:

\begin{center}
    $\aPath, \aPathb$ = $\epsilon$ $|$ $\labelComodin$ $|$ $\aLabel$ $|$ $\aLabel^{-}$ $|$ $\expNodoEnCamino{\aFormula}$ $|$ $\aPath$ . $\aPathb$ $|$ $\aPath \pathUnion \aPathb$ $|$ $\aPath \pathIntersection \aPathb$ $|$ $\aPath^{*}$ $|$ $\pathComplement{\aPath}$ $|$
    $\aPath^{n,m}$ 
\end{center}  

\noindent where $\aLabel$ iterates over all labels from $\Sigma_e$ and $\aFormula$ is a \textit{node expression} defined by the following grammar:

\begin{center}
    $\aFormula, \aFormulab$ = $\neg \aFormula \mid \aFormula \wedge \aFormulab$ $|$ $\comparacionCaminos{\aPath}$ $|$ $\esDatoIgual{\aData}$ $|$ $\esDatoDistinto{\aData}$ $|$ $\comparacionCaminos{\aPath = \aPathb}$ $|$ $\comparacionCaminos{\aPath \neq \aPathb}$ $|$
    $\aFormula \vee \aFormulab$
\end{center}

\noindent where $\aPath$ and $\aPathb$ are path expressions (and so they are defined by mutual recursion), and $c$ iterates over $\Sigma_n$. If we only allow the Kleene star to be applied to labels and their inverses (i.e., $\aLabel^-$) we then have a subset of \Gregxpath called \Gcorexpath. The semantics for these languages are defined in~\cite{libkin2016querying} in a similar fashion as the usual regular languages for navigating graphs,~\cite{barcelo2013querying}, while adding some extra capabilities such as the complement of a path expression $\pathComplement{\aPath}$ and data tests. The $\comparacionCaminos{\aPath}$ operator is the usual one for \textit{nested regular expressions} or NREs used in \cite{barcelo2012relative}. Given a data-graph $\aGraph=(V,L,D)$, the semantics is defined as follows:

\begin{itemize}[leftmargin=.6in,label={}]

\itemsep0em 

  \item $\semantics{\epsilon}_\aGraph = \{(v,v)$ $|$ $v \in V\}$
  
  \item $\semantics{\labelComodin}_\aGraph = \{(v,w)$ $|$ $v,w \in V, L(v,w) \neq \emptyset\}$ 
  
  \item $\semantics{\aLabel}_\aGraph = \{ (v,w)$ $|$ $\aLabel \in L(v,w)\}$
  
  \item $\semantics{\aLabel^-}_\aGraph = \{ (w,v)$ $|$ $\aLabel \in L(v,w)\}$
  
  \item $\semantics{\aPath^*}_\aGraph = $ the reflexive transitive closure of $\semantics{\aPath}_\aGraph$
  
  \item $\semantics{\aPath . \aPathb}_\aGraph = \semantics{\aPath}_\aGraph \circ \semantics{\aPathb}_\aGraph$
  
  \item $\semantics{\aPath \pathUnion \aPathb}_\aGraph = \semantics{\aPath}_\aGraph \cup \semantics{\aPathb}_\aGraph$
  
   \item $\semantics{\aPath \pathIntersection \aPathb}_\aGraph = \semantics{\aPath}_\aGraph \cap \semantics{\aPathb}_\aGraph$
    
  \item $\semantics{\pathComplement{\aPath}}_\aGraph = V \times V \setminus \semantics{\aPath}_\aGraph$
  
  \item $\semantics{\expNodoEnCamino{\aFormula}}_\aGraph = \{(v,v)$ $|$ $v \in \semantics{\aFormula}_\aGraph$\}
  
  \item $\semantics{\aPath^{n,m}}_\aGraph = \bigcup\limits_{k=n}^m(\semantics{\aPath}_\aGraph)^k$
  
  \item $\semantics{\comparacionCaminos{\aPath}}_\aGraph = \pi_1(\semantics{\aPath}_\aGraph) = \{v \mid \exists w\in V\, (v,w) \in \semantics{\aPath}_\aGraph \} $ 
  
  \item $\semantics{\lnot \aFormula}_\aGraph = V \setminus \semantics{\aFormula}_\aGraph$
  
  \item $\semantics{\aFormula \wedge \aFormulab}_\aGraph = \semantics{\aFormula}_\aGraph \cap \semantics{\aFormulab}_\aGraph$
  
  \item $\semantics{\aFormula \vee \aFormulab}_\aGraph = \semantics{\aFormula}_\aGraph \cup \semantics{\aFormulab}_\aGraph$
  
  \item $\semantics{\esDatoIgual{c}}_\aGraph = \{v \in V$ $|$ $D(v) = \esDato{c}$\}
  
  \item $\semantics{\esDatoDistinto{c}}_\aGraph = \{v \in V$ $|$ $D(v) \neq \esDato{c}$\}
  
  \item $\semantics{\comparacionCaminos{\aPath = \aPathb}}_\aGraph = \{v \in V$ $|$ $\exists v', v'' \in V$, $(v,v') \in \semantics{\aPath}_\aGraph$, $(v,v'') \in \semantics{\aPathb}_\aGraph$, $D(v') = D(v'')$\} 
  
  \item $\semantics{\comparacionCaminos{\aPath \neq \aPathb}}_\aGraph = \{v \in V$ $|$ $\exists v', v'' \in V$, $(v,v') \in \semantics{\aPath}_\aGraph$, $(v,v'') \in \semantics{\aPathb}_\aGraph$, $D(v') \neq D(v'')$\} 
  
\end{itemize}

We use \defstyle{$\aPath \entoncesCamino \aPathb$} to denote the path expression $\aPathb \pathUnion \pathComplement{\aPath}$, and \defstyle{$\aNodeExpression \entoncesNodo \aNodeExpressionb$} to denote the node expression $\aNodeExpressionb \lor \neg \aNodeExpression$. We also note a label \defstyle{$\aLabel$} as $\down_\esLabel{\aLabel}$ in order to easily distinguish the `path' fragment of the expressions. For example, the expression $\esLabel{son\_of} [\esDatoIgual{Maria}] \esLabel{sister\_of}$ will be noted as $\down_\esLabel{son\_of} [\esDatoIgual{Maria}] \down_\esLabel{sister\_of}$.

Naturally, the expression $\aPath \cap \aPathb$ can be rewritten as $\overline{\overline{\aPath} \cup \overline{\aPathb}}$ while preserving the semantics, and the same holds for operators $\wedge$ and $\vee$ for the case of node expressions using the $\lnot$ operator. In this grammar, all these operators are defined given that there are some natural restrictions of the language which have the same grammar except for some ``negative'' productions (like $\overline{\aPath}$), in which these simulations would not be possible, and therefore we need to consider them separately in the definition. Considering only the ``positive'' fragment of a language usually allows obtaining better complexity bounds in reasoning problems. This happens mainly because of the \textit{monotony} property: intuitively, a set of expressions from a language $\mathcal{L}$ over structures $\mathcal{H}$ is said to satisfy monotony if whenever $H \in \mathcal{H}$ satisfies $\nu \in \mathcal{L}$, then $H'$ satisfies $\nu$ for every $H \subseteq H'$.

In our case, \defstyle{\Gposregxpath} is defined as the subset of \Gregxpath expressions that do not use $\overline{\aPath}$ nor $\lnot \aFormula$. Thus, in \Gposregxpath we will not be able to simulate the $\cap$ operator unless it is present in the original \Gregxpath grammar. Moreover, it can be shown that the set of \Gposregxpath node expressions satisfy monotony in the following sense: if $v \in \semantics{\nu}_\aGraph$ for $\nu \in \Gposregxpath$, then $v \in \semantics{\nu}_{\aGraph'}$ for every data-graph $\aGraph' \supseteq \aGraph$.

Given a $\Gregxpath$ formula $\eta$, we denote by \defstyle{$\Sigma_n^\eta$} the set of data values that are mentioned in $\eta$. More precisely, $\Sigma_n^\eta = \{\esDato{c} \in \Sigma_n : \esDatoIgual{c} \text{ or } \esDatoDistinto{c} \text{ is a subexpression of }\eta\}$. Note that $|\Sigma_n^\eta| \leq |\eta|$.


\section{Data cleaning} \label{section:dataCleaning}

One of the intuitive ideas behind an EMDG $(\aProbabilisticDatabase,\aRealizationModel)$ is that $\aProbabilisticDatabase$ represents a known distribution of `possible worlds', while $\aRealizationModel$ represents how each of them could be `observed' in a noisy manner with the potential introduction of errors.
A PUDG $\aPUDG = (\aProbabilisticDatabase,\aRealizationModel,\aGraph')$ adds a concrete observation (namely $\aGraph'$) to our model, allowing us to actually reason about the `underlying reality' given our knowledge of possible worlds and the limits or deficiencies of our observation method. A central question that arises naturally in this interpretation, already mentioned in Example~\ref{example:BibliographicalRIntroducesAuthorship}, is the functional problem of \functionProblem{Data cleaning}, which aims to find the most probable clean world (i.e., a data-graph in $\aProbabilisticDatabase$) given the epistemic model of reality and the observation.

\begin{definition} [\functionProblem{Data cleaning}]
We define:
\begin{center}
\fbox{\begin{minipage}{30em}
  \textsc{Problem}: Most likely data-graph 

\textsc{Input}: A $\Pi$ PUDG $\aPUDG = (\aProbabilisticDatabase,\aRealizationModel,\aGraph')$

\textsc{Output}: A data-graph $I \in \aProbabilisticDatabase$ such that the probability of  $\aInverseRealizationModel_{\aProbabilisticDatabase, \aRealizationModel}(\aGraph')(I)$ is maximum. 
\end{minipage}}
\end{center}
\end{definition}


Recall that this data-graph with maximal probability exists since, by definition of $\aRealizationModel$, for a fixed $\aGraph'$, $\aInverseRealizationModel_{\aProbabilisticDatabase, \aRealizationModel}(\aGraph')$ is a probabilistic data-graph, and therefore there is a finite number of graphs with probability greater than~0. Note however that, if we consider unrestricted $\aProbabilisticDatabase$ or $\aRealizationModel$, then the problem might not be computable.  

We also define a decision problem related to data cleaning, \decisionProblem{data cleaning lower bound}:


\begin{center}
\fbox{\begin{minipage}{30em}
  \textsc{Problem}: Is there a sufficiently likely clean data-graph? 

\textsc{Input}: A $\Pi$ PUDG $\aPUDG = (\aProbabilisticDatabase,\aRealizationModel,\aGraph')$ and a (rational) bound $b \in [0,1]$

\textsc{Output}: Decide whether there exists a data-graph $I \in \aProbabilisticDatabase$ such that $\aInverseRealizationModel_{\aProbabilisticDatabase, \aRealizationModel}(\aGraph')(I) > b$. 
\end{minipage}}
\end{center}

Note that an algorithm for solving the data cleaning problem can be used to solve the decision version by checking whether the data-graph $I$ found by the algorithm satisfies $\aInverseRealizationModel_{\aProbabilisticDatabase, \aRealizationModel}(\aGraph')(I) > b$.

In what follows we study this problem considering Subset, Superset and Update PUDGs, respectively. For each of them we show that the problem proves intractable even when considering models that allow to efficiently compute \aInverseRealizationModel. Nonetheless, we also find some tractable restrictions that bound the number of data-graphs which derive into the unclean one.

\subsection{Data cleaning in Subset PUDG}

We start studying the data cleaning problem for Subset PUDGs. Suppose that we have a Subset PUDG $\aPUDG = (\aProbabilisticDatabase, \aRealizationModel, \aGraph')$, and we want to find the data-graph $I \in \aProbabilisticDatabase$ that maximizes $\aInverseRealizationModel(\aGraph')(I)$. Note that it follows from Definition~\ref{def:PUDG} that the cosupport of $\aGraph'$ is finite, however this condition does not prevent that its cardinal might be of arbitrary size, and thus generating it could result too expensive. Moreover, even if we assume that $\aProbabilisticDatabase$ and $\aRealizationModel$ can be computed in polynomial-time on the input size, this is not enough to deduce that $\aInverseRealizationModel$ is also efficiently computable. We prove that even in the case where the cosupports and $\aInverseRealizationModel$ are efficiently computable, the problem can still be intractable as the next theorem shows.

\begin{theorem}
\label{teo:subsetNpComplete}
There is a fixed EMDG $(\aProbabilisticDatabase,\aRealizationModel)$ with $\aProbabilisticDatabase$,  $\aRealizationModel$ and $\aInverseRealizationModel_{\aProbabilisticDatabase, \aRealizationModel}$ efficiently computable, for which the problem \decisionProblem{data cleaning lower bound} is \textsc{NP-complete}, even when considering Subset PUGDs with no node deletions (i.e., $\aRealizationModel(\aGraph)(\aGraph') > 0$ implies $V_\aGraph = V_\aGraph'$).
\end{theorem}

\begin{proof}

The upper bound can be easily proven by noticing that, since we are not allowing node deletions, if there is an $I$ such that $\aInverseRealizationModel(\aGraph')(I) > b$ then $|I| = poly(|\aGraph'|)$ and $\aInverseRealizationModel(\aGraph')(I)$ can be computed in $poly(|\aGraph'|+|I|)$. Now, for the hardness part, we reduce 3-SAT to Subset PUGD with a fixed $\aProbabilisticDatabase$ and $\aRealizationModel$.

Intuitively, $\aProbabilisticDatabase$ will be a distribution over all data-graphs $\aGraph_{\phi,v}$ representing a Boolean formula $\phi$ alongside an assignment $v$ over the variables of $\phi$. 
Probability $\aProbabilisticDatabase(\aGraph_{\phi,v})$ will be defined so that it is higher when $v\models\phi$. On the other hand, $\aRealizationModel$ will be defined in such a way that
it will simply delete the part of the graph that represents the assignment. Then, given a data-graph $\aGraph'_\phi$ 
representing a formula without an assignment, the data cleaning problem will consist in finding, if there exists, the assignment $v$ that makes $\phi$ evaluate to true. 

Formally, let $\Sigma_e=\{\esLabel{is\_literal},\esLabel{is\_literal\_negated},\esLabel{value},e_1,e_2\}$ and $\Sigma_n=\{\esDato{var},\esDato{clause},\top,\bot\}$. For every CNF formula $\phi$ of $n$ variables $x_1,\ldots,x_n$ and $m$ clauses $c_1,\ldots,c_m$ and 
every assignment $v:var(\phi) \to \{\bot, \top\}$, 
the representation of $\phi$ and $v$ is the data-graph 
$\aGraph_{\phi, v}=(V,L,D)$, given by:
\begin{enumerate}
    \item $V=\{1,2,\ldots,n+m+2\}$. Elements from $1$ to $n$ represent the variables, elements from $n+1$ to $n+m$ represent the clauses, and $n+m+1,n+m+2$ are two special nodes. From now on, we refer to the $i$th $\esDato{var}$ node as the variable $x_i$, 
    and to the $j$th $\esDato{clause}$ node as the $c_j$ clause.
    
    \item $D(x_i)=\esDato{var}$ for every $i=1,\ldots,n$, $D(c_j)=\esDato{clause}$ for every $j=1,\ldots,m$, $D(n+m+1)=\bot$ and $D(n+m+2)=\top$.

    \item For each variable $x_i$ and clause $c_j$, $L(x_i,c_j)=\{\esLabel{is\_literal}\}$ if $x_i$ is a literal of $c_j$, $L(x_i,c_j)=\{\esLabel{is\_literal\_negated}\}$ if $\lnot x_i$ is a literal of $c_j$.
    
    \item\label{red:observed} For each variable $x_i$, $L(x_i,n+m+1)=\{\esLabel{value}\}$ if $v(x_i) = \bot$, or $L(x_i,n+m+2)=\{\esLabel{value}\}$ if $v(x_i) = \top$. 
    
    \item\label{red:topbot} $L(n+m+1,n+m+2)=\{e_i\}$ for some $i=1,2$. 
    \item The data-graph has no other edges.
\end{enumerate}

Notice that the assignment $v$ is only relevant in rule (\ref{red:observed}). If we dropped this rule and (\ref{red:topbot}) we obtain for each $\phi$ a unique data-graph representation denoted by $G'_\phi$. 

Given an arbitrary data-graph $\aGraph$ it is possible to decide in polynomial time if $\aGraph$ has the form $\aGraph_{\phi, v}$ and obtain in polynomial time a formula $\phi$ and an assignment $v$ that it represents. Considering this, we define a distribution across the set of data-graphs in the following way:

\[ \mathcal{I}(\aGraph) = \begin{cases} 
     \frac{1+\phi_v}{N(\phi)} & \text{if }\aGraph = \aGraph_{ \phi, v} \text{ for some } \phi, v \text{ and } (\bot, e_1, \top) \in E_\aGraph \\
      \frac{1-\phi_v}{N(\phi)} & \text{if }\aGraph = \aGraph_{ \phi, v} \text{ for some } \phi, v \text{ and } (\bot, e_2, \top) \in E_\aGraph\\
      0 & \text{otherwise},
   \end{cases}
\]
where $N(\phi)$ is a normalization factor based on the number of variables, clauses and assignments of $\phi$, that allows the probability function to satisfy $\sum_{I \in U} \aProbabilisticDatabase(I) = 1$ and $\phi_v$ is either 1 or 0, depending on whether $v\models\phi$ or not, respectively. We now show a way to build this normalization factor so that it can be easily computed given a formula $\phi$.

First, we define $N(\phi)$ such that, given some natural numbers $n$ and $m$, the total probability assigned to the set of data-graphs representing formulas with $n$ variables and $m$ clauses is precisely the entry in column $n$ and row $m$ in Table~\ref{table:cantor}. We denote from now on this entry of the table by $T(n,m)$. Note that, if the probabilities are assigned in this way, then $\aProbabilisticDatabase$ is a well-defined probabilistic database. Also, any entry $T(n,m)$ can be computed in polynomial-time on $n$ and $m$.\footnote{It can be shown that the entry in row $m$ and column $n$ has value $2^{-(1 + \frac{n(n+1)}{2} + m'(n'+1) + \frac{m'(m' + 1)}{2})}$}

Now, we will distribute the probability evenly across all data-graphs that represent a formula of $n$ variables and $m$ clauses, and thus the final probability is obtained by multiplying this number by the corresponding value in the table. The final probability assigned is therefore $T(n,m) \times \frac{1}{C(n,m)}$, where\footnote{The formula can be derived in the following way: we need to account for the $2^n$ possible assignations of the $n$ variables, all the $8 \binom{n}{3}$ ways in which each clause can independently pick the literals that appear in it, and finally the 2 distinct labels that the edge between $\bot$ and $\top$ can have.} $C(n,m) = 2 \times 2^n (8\binom{n}{3})^m$.

\begin{table}[H]
\begin{tabular}{|l||ccccc|}
\hline
\diagbox{$m$}{$n$} & 0        & 1        & 2        & 3    & $\cdots$    \\
\hline\hline \rule{0pt}{15pt}
0   & $\frac{1}{2^1}$ & $\frac{1}{2^2}$ & $\frac{1}{2^4}$ & $\frac{1}{2^7}$ & \\
\rule{0pt}{15pt}
1   & $\frac{1}{2^3}$ & $\frac{1}{2^5}$ & $\frac{1}{2^8}$ &         &  \\ \rule{0pt}{15pt}
2   & $\frac{1}{2^6}$ & $\frac{1}{2^9}$ &          &     $\ddots$  & \\\rule{2pt}{0pt}
$\vdots$ &&&&& \\ 
\hline
\end{tabular}
\vspace*{5mm}
\caption{Table showing how much of the ``probability weight'' is assigned into each set of data-graphs, for each value of $n$ and $m$.}
\label{table:cantor}
\end{table}

Regarding $\aRealizationModel$, it maps each data-graph $\aGraph$ to a uniform distribution across every data-graph $H$ with the same nodes as $\aGraph$ but just a subset of its edges. That is, $\aRealizationModel(\aGraph)(H) = \frac{1}{2^{|E_\aGraph|}}$ if $V_H = V_\aGraph$ and $H \subseteq \aGraph$, and $\aRealizationModel(\aGraph)(H) = 0$ otherwise.

Finally, the reduction is as follows: given a 3-CNF formula $\phi$ of $n$ (ordered) variables and $m$ (ordered) clauses, we construct the data-graph $\aGraph'=\aGraph'_\phi$ and define $b = \frac{1}{2^{|E_{\aGraph'}|}} \times T(n, m) \times \frac{1}{C(n, m)}$. Clearly, the data-graph can be built in $poly(|\phi|)$ time, and each element in the product that composes $b$ can also be built in time $poly(n+m)$. It is clear that if there exists an assignment $v$ to the variables of $\phi$ that evaluates it to true, then $\aInverseRealizationModel(\aGraph'_\phi)(\aGraph_{\phi,v}) > b$, otherwise no data-graph in the cosupport of $\aGraph'_\phi$ will have a probability that is bigger than $b$.

\end{proof}

This theorem shows that even if we fix both functions $\aProbabilisticDatabase$ and $\aRealizationModel$ the problem remains intractable. Note that fixing this part of the input seems reasonable: both functions $\aProbabilisticDatabase$ and $\aRealizationModel$ capture domain knowledge and the semantics of the data the data-graph represents.

 There are, however, at least two ways in which we can simplify this problem to obtain a polynomial-time resolution.
 
 First, we could try and control the complexity that arises from the relationship between $\aGraph$ and the EMDG. It might be the case that we find an interesting restriction of the problem by considering some EMDG that captures topological aspects of the data and a particular subclass of data-graphs. For example, in~\cite{barcelo2017data} some database repairing problems are conceptually similar to the ones discussed here, and can be solved in polynomial time when considering graphs with bounded treewidth. 
 
 Second, we could control the complexity encapsuled in both $\aProbabilisticDatabase$ and $\aRealizationModel$ or by restricting the overall possibilities of those functions as input of the problem.
 
For instance, if we condition the problem to obtain a polynomial bound on the cardinality of the data-graphs that are able to ``transitition'' to the observed state $\aGraph$ through $\aRealizationModel$, then it becomes tractable under the assumption that $\aInverseRealizationModel$ is efficiently computable. More precisely, we ask for the size of the cosupport of a data-graph $\aGraph'$, namely $\sigma(\aGraph')$, to be bounded by some polynomial function. We can do this by forbidding node deletions and by bounding the number of edges that can be added by a constant $k_e$. Formally, we assume that $\aRealizationModel(I)(\aGraph') = 0$ if $|E_I \setminus E_\aGraph'| > k_e$ or $V_I \neq V_{\aGraph'}$. Assuming these restrictions, we can show that $\sigma(\aGraph')$ is bounded by a polynomial on $n=|V_{\aGraph'}|$:
$$
|\sigma(\aGraph')| \leq \sum_{\substack{i=0 \\ i\leq \missing{\aGraph'}}}^{k_e} \binom{\missing{\aGraph'}}{i} \leq k_e (\missing{\aGraph'})^{k_e} = \bigO{n^{2 k_e}},
$$
where $\missing{\aGraph'} \eqdef |\Sigma_e| n^2 - |E_{\aGraph'}|$ counts the number of edges missing from $\aGraph'$, and therefore the ones that could have been deleted.

The next theorem 
follows immediately from the previous discussion. 

\begin{theorem}\label{teo:subset-bounded}
\functionProblem{Data cleaning} of Subset PUDG can be computed in polynomial time in the size of the observed data-graph $\aGraph'$ whenever the following conditions hold: (i) node deletions are not allowed, and (ii) the number of edges deleted is bounded by a constant.
\end{theorem}


Observe that, even though bounding the number of possible deletions is a considerable restriction to the problem, the resulting family of realization models are quite expressive, as they can assign probabilities based on topological properties and on the different present data values and edge labels.

\subsection{Data cleaning in Superset PUDG}

Now we consider the Superset PUDG version of the Data Cleaning problem. This is, given a Superset PUDG $\aPUDG = (\aProbabilisticDatabase, \aRealizationModel, \aGraph')$, we want to find the data-graph $I \in \aProbabilisticDatabase$ that maximizes $\aInverseRealizationModel_{\aProbabilisticDatabase, \aRealizationModel}(\aGraph')(I)$.
Recall that $\aGraph'$ is the observed unclean data-graph that may have loops and as many edges as distinct possible labels. Notice that, in this case, and disregarding whether we consider node and/or edge deletions, every graph in $I$ and in the cosupport of $\aGraph'$ is necessarily a subgraph of $\aGraph'$. Therefore, the set of data-graphs in $\aInverseRealizationModel(\aGraph')$ is finite, and its size is bounded by $\sum_{i=0}^n \binom{n}{i} 2^{i^{2}\times|\Sigma_e|} \leq 2^{n^{2}\times|\Sigma_e|+n}$ ($n=|V_{\aGraph'}|$) in the worst case scenario. 
Moreover, if we only consider edge deletions, then $\sigma(\aGraph')$ is bounded by $2^{n^{2}\times|\Sigma_e|}$.

As in the case of Subset PUDGs, the fact that $\aInverseRealizationModel_{\aProbabilisticDatabase, \aRealizationModel}$ is efficiently computable is not enough to conclude that we can find in polynomial time the most probable clean data-graph:

\begin{theorem}\label{teo-superset-fixed}
There exists a fixed $\aProbabilisticDatabase$ and $\aRealizationModel$ such that $\aInverseRealizationModel_{\aProbabilisticDatabase, \aRealizationModel}$ is efficiently computable the problem Data Cleaning of Superset PUDGs is \textsc{NP-complete}, even when $\aRealizationModel$ does not allow node additions
\end{theorem}

\begin{proof}
The proof is the same one used for Theorem \ref{teo:subsetNpComplete}. See the appendix for the details.
\end{proof}

As before, this result tells us that even though the search space is now clearly bounded we will need to impose some further restrictions in the structure of both $\aRealizationModel$ and $\aProbabilisticDatabase$.

If we bound the number of nodes and edges that can be added through $\aRealizationModel$ by some constant then the problem is tractable, assuming an efficiently computable $\aInverseRealizationModel$.

\begin{theorem}\label{teo:supersetBoundedByC}
    The Data Cleaning Superset PUDG problem can be solved in polynomial time if $\aRealizationModel(\aGraph)(\aGraph') > 0$ implies $|V_{\aGraph'} \setminus V_{\aGraph}| + |E_{\aGraph'} \setminus E_{\aGraph}| < c$ for some fixed $c\in\N$ and any pair of data-graphs $\aGraph, \aGraph'$
\end{theorem}

\begin{proof}
This can be proven in the same way as Theorem~\ref{teo:subset-bounded}. There is only a polynomial number of data-graphs that can be mapped to a positive probability through $\aRealizationModel$ given the observed data-graph $\aGraph$, and they can be easily enumerated.
\end{proof}

Observe that this theorem shows that in the superset case we can allow $\aRealizationModel$ to add nodes and still have a polynomial algorithm for solving the data cleaning problem. In the analogous situation of the subset case, we could not allow to remove nodes because when trying to explore the cosupport through $\aRealizationModel$ we would need to guess the data values of the nodes that were deleted, and the possibilities for these data values aren't bounded. 

\subsection{Data Cleaning in Update PUDGs}

As in the previous cases, if we consider efficiently computable $\aInverseRealizationModel$'s, then the problem can still be intractable:

\begin{theorem}\label{teo-update-fixed}
There exists a fixed $\aProbabilisticDatabase$ and $\aRealizationModel$ such that $\aInverseRealizationModel_{\aProbabilisticDatabase, \aRealizationModel}$ is efficiently computable and the problem Data Cleaning of Update PUDGs is \textsc{NP-complete}, even when $\aRealizationModel$ only allows updating edge labels.
\end{theorem}

\begin{proof}
Again, the proof on this theorem is similar to the one used for Theorem~\ref{teo-superset-fixed}. See the appendix for the details.
\end{proof}

As before, we can bound the size of the cosupport of  $\aGraph'$ in order to deduce a restriction of the problem that can be solved in polynomial time. We can do this in a natural way by bounding the number of nodes whose data values could be changed by $\aRealizationModel$, in a similar way as in Theorem~\ref{teo:supersetBoundedByC}. Moreover, for every data value $\esDato{c}$, we need a bound on the number of data values that can be updated to $\esDato{c}$ through $\aRealizationModel$. We can formalize this by requiring the existence of a constant $k$ and a \textbf{$k$-data-prior} function $f:\Sigma_n \to \mathcal{P}(\Sigma_n)$ computable in polynomial time on the representation of the input data value such that $|f(\esDato{c})| \leq k$, for every $\esDato{c} \in \Sigma_n$. Then, given a function $f$ with the previous restrictions, we can consider Update EMDGs $(\aProbabilisticDatabase, \aRealizationModel)$ such that $\aRealizationModel$ agrees with $f$ in the following way: if $\aRealizationModel(\aGraph)(\aGraph') > 0$, then for every $v \in V_{\aGraph'}$ holds that $D_\aGraph(v) \in f(D_{\aGraph'}(v))$. Hence, we can prove the following result:

\begin{theorem}\label{teo:update-bounded}
Given a $k$-data-prior function $f$ as input such that $\aRealizationModel$ depends on $f$, if $\aRealizationModel$ is allowed to update only a constant number $z$ of data values, then the Data Cleaning Node Update PUDG problem can be solved in polynomial time. 
\end{theorem}

\begin{proof}
    Given the observed unclean data-graph $\aGraph'$, we generate the cosupport of $\aGraph'$ given $\aRealizationModel$ using the $k$-data-prior function $f$. In particular, the size of the cosupport is bounded by $\binom{n}{z} k^z = \bigO{n^z}$, and it can be generated in polynomial time on its size. Then, we only need to evaluate $\aInverseRealizationModel$ and keep the most probable one.
\end{proof}

\subsection{Global Cardinality Constraints}

We now focus on some other restrictions of the Data Cleaning problem that are particularly interesting within the framework of data-graphs as they allow modifying the data values on the nodes. Therefore, in this section we consider Node-Update PUDGs.

Suppose that we have a prior on how many times each edge label $\esLabel{e}$ or data value $\esDato{c}$ is present in the original data-graph. In this scenario, we consider as a preferred answer those data-graphs that are closer to the prior.
In other words, we consider cardinality constraints surrounding the number of appearances of one or more predetermined data subsets of values and/or edge labels.

First, we present some tractable versions based on cardinality constraints. Afterward, we explore other restrictions to the problem that make it remain hard. 
This allows us to shed light on the line between tractability and hardness for Node-Update EMDG. 

\subsubsection{\nombreRooksSection}
Let us consider the following restriction of the general Node-Update EMDG model, where data in nodes are unique and belong to a predefined subset of ``valid'' data values. Let $\aGraph'=(V,L,D)$ be our observed unclean data-graph, suppose that $|\Sigma_n|= |V| = N$ and $D(V)=\Sigma_n$
, and that $\Sigma_n$ is part of the input. 
Let $m<N$, and let $f_{\aGraph'}: [1, m] \to V$ be an injective function that indicates for which nodes of the observed graph we know with certainty that their data-values are correct. More precisely, $\aInverseRealizationModel(\aGraph')(I) = 0$ if $I$ modifies a node in the image of $f$. In other words, we have certainty over the data values for precisely $m$ nodes of $\aGraph'$, whereas for the remainder of the nodes we do not.
Suppose we are given a weight function such that these weights determine the probability of a certain assignment of data-values to each node. More precisely, let $w: \Sigma_n \times V \rightarrow \left[0,1\right]$ such that $\sum_{I \in \aProbabilisticDatabase} \prod^{N}_{i,j=1} w_{ij} p^{(I)}_{ij} = 1$, where 
$$ 
p^{(I)}_{ij} = \begin{cases} 1 & \mbox{if the data-value } i \mbox{ has been assigned to the node } j \\
0 & \mbox{if not}
\end{cases}
$$
and $\sum^{N}_{i=1} p^{(I)}_{ij} = 1$ for every $j=1, \ldots, N$ for every $I\in\aProbabilisticDatabase$. Moreover, suppose that $\aRealizationModel(H)(G')=1$ for every $H \in \aProbabilisticDatabase$. Considering all these assumptions, the value $\aInverseRealizationModel(\aGraph')(I)$ is: 
\begin{small}
$$
\aInverseRealizationModel(\aGraph')(I) = \frac{\aProbabilisticDatabase(I) \times \aRealizationModel(I)(\aGraph')}{\sum \aProbabilisticDatabase(H) \times \aRealizationModel(H)(\aGraph')} = \frac{\aProbabilisticDatabase(I)}{\sum \aProbabilisticDatabase(H)} = \aProbabilisticDatabase(I)= \prod_{(i,j) \in \Sigma_n \times (V \setminus Im(f))} w_{ij} p^{(I)}_{ij}.
$$
\end{small}

Since we are looking for the graph with the maximum probability, this is equivalent to finding those weights for which the product is maximum.
This allows us to take into consideration the adjacencies and labels in our observed unclean data-graph when assigning the weights.
Additionally, we could also ask for an extra feature assuming that $\sum^{N}_{j=1} w_{ij} = 1$ for every $j=1, \ldots, N$. This condition can be interpreted as: ``all the possible data value assignments for each node determine a probability by itself''.
In summary, we define the problem \nombreRooks\  as follows: 

\begin{center}
\fbox{\begin{minipage}{25em}
  \textsc{Problem}: \nombreRooks
  
\textsc{Input}: A set of data-values $\Sigma_n$, a PUDG $\aPUDG = (\aProbabilisticDatabase,\aRealizationModel,\aGraph')$ over the set of data-values $\Sigma_n$, an injective function $f_{\aGraph'}: [1, m] \to V_{\aGraph'}$, and a weight function $w: \Sigma_n \times V(\aGraph') \rightarrow \left[0,1\right]$ as described above.

\textsc{Output}: A data-graph $I_m \in \aProbabilisticDatabase$ such that the probability of  $\aInverseRealizationModel(\aGraph')(I_m)$ is maximum. 
\end{minipage}}
\end{center}

\begin{theorem}
The problem {\nombreRooks} can be solved in $\bigO{n^3}$~time.
\end{theorem}

\begin{proof}
To determine the complexity of this problem, we can reinterpret it as an $N$-rooks problem as follows. 
Consider an $N$-square board, where the ranks (i.e.\ the rows) represent the data values and the files (i.e.\ the columns) represent the nodes of $\aGraph'$. 
We place a rook in square $(i,j)$ if we are certain that the node $j$ has the data value $i$. Thus, there are $m$ rooks with a fixed position in our board, all of them in non-attacking positions since we assumed that these $m$ values are distinct and assigned to different nodes. Now, consider that our board has assigned weights $w_{ij}$ to each square, 
which represents the probability of the assignment of each data value to each node.
We can determine the remaining positions of the $N-m$ rooks by considering an auxiliary graph $H$ as follows. Let $H$ be a graph that has its vertices partitioned into $V_1$ and $V_2$, having $N$ vertices each: one of them --let us say $V_1$-- has one vertex for each data value, and the other one has one vertex for each node. There are no edges between vertices in the same partition, and there is an edge between $x_i$ in $V_1$ and $y_j$ in $V_2$ if and only if the data value $i$ can be assigned to the node $j$ with positive probability. More precisely, we consider that the edge $x_i y_j$ has a weight $w_{ij}$, this is the weights of the edges between vertices in $V_1$ and $V_2$ are now the probabilities determined on each square.
Hence, solving Node-Update in this case derives from finding a maximal weighted matching in the bipartite graph $H'$, which can be done in $\bigO{n^3}$~time~\cite{haynes1998domination}. 
\end{proof}

\subsubsection{Weak cardinality constraint over data values}\label{section:localAndGlobal}

We could consider an alternative scenario in which the function $\aProbabilisticDatabase$ codifies a cardinality preference criteria by having a function $T: U \times \Sigma_n \to \N_0$ such that $\sum_{\esDato{c} \in \Sigma_n}T(\aGraph, \esDato{c}) = |V_\aGraph|$ for all $\aGraph$. This function $T$ tells us the number of data values that we want to have for each data value $\esDato{c}$. Naturally, $T(\aGraph, \esDato{c})$ will be bigger than $0$ for at most $|V_\aGraph|$ different data values. Given a function $T$ and a data-graph $H$ we define the penalization factor of $H$ with respect to $T$ as

$$
d_T(H) = \sum_{\esDato{c} \in \Sigma_n} |T(H, \esDato{c}) - count(H, \esDato{c})|
$$
where $count(H, \esDato{c})$ denotes the number of times the data value $\esDato{c}$ is present in $H$. We say that \emph{$\aProbabilisticDatabase$ depends on $T$} if for every pair of data-graphs with the same set of nodes it is true that $\aProbabilisticDatabase(H) > 0 $ if and only if $d_T(H) = 0$ and $\aProbabilisticDatabase(H)>0,\aProbabilisticDatabase(H') > 0$ implies $\aProbabilisticDatabase(H) = \aProbabilisticDatabase(H')$. This basically means that $\aProbabilisticDatabase$, given a set of nodes $V$, distributes the probability evenly accross all data-graphs with a set of nodes that also satisfy the constraints defined by the function $T$.

Notice that this prioritizes data-graphs that follow a cardinality criterion over the set $\Sigma_n$, completely ignoring the uncleaned observed data-graph $\aGraph$. We add a simple local realization model $\aRealizationModel$ between clean and unclean data-graphs by using a function $\delta: \Sigma_n \times \Sigma_n \to \N_0$ that defines, for each pair of data values $(\esDato{c},\esDato{d}) \in \Sigma_n^2$, the `cost' of changing $\esDato{c}$ to $\esDato{d}$. Naturally, we ask that $\delta(\esDato{c},\esDato{c}) = 0$ for every $\esDato{c} \in \Sigma_n$. In other words, $\aRealizationModel$ prioritizes those data-graphs $H$ such that the cost of the data values transitions required to go from $H$ to the unclean data-graph $\aGraph$ is minimized. 

More formally, given two data-graphs $H$ and $\aGraph$ such that $V_H=V_\aGraph, L_H=L_\aGraph$ and a weight function $\delta$ between data values, we define the \defstyle{cost} of going from $H$ to $\aGraph$ as
$$
c_\delta(H,\aGraph) = \sum_{v \in V_H} \delta(D_H(v),D_\aGraph(v)).
$$

We say that \emph{$\aRealizationModel$ depends on $\delta$} if $\aRealizationModel(H)(\aGraph) > \aRealizationModel(H')(\aGraph')$ if and only if $c_\delta(H,G) < c_\delta(H',G')$.

Observe that the following theorem is (almost) a generalization of the scenario described in the previous section. In that case, the weight $w$ of the transition between data values could also depend on the involved vertex.

\begin{theorem} \label{Theorem:WeakCardinalityConstraints}
The \functionProblem{Data cleaning} problem can be solved in polynomial time if $\aProbabilisticDatabase$ depends on a given function $T$ and $\aRealizationModel$ depends on a given local function $\delta$.\footnote{Here, we assume that $T$ and $\delta$ are functions that can be computed in polynomial time.}

\end{theorem}

\begin{proof}
Given the Update PUDG $(\aProbabilisticDatabase,\aRealizationModel, \aGraph)$ where $\aProbabilisticDatabase$ depends on $T$ and $\aRealizationModel$ depends on $\delta$, we know that the most probable clean data-graph will be the one satisfying the constraints defined by $T$ and that minimizes the cost of the transitions as defined in $\delta$. We define the complete bipartite undirected graph $I=(V_I,E_I)$ in such a way that a minimum weighted perfect matching in $I$ codifies the cleanest data-graph~$H$:
\begin{align*}
& V_I  = \{v : v \in V_\aGraph\}  \cup 
\{u_{\esDato{c},i} : \esDato{c} \in \Sigma_n, T(\aGraph,\esDato{c}) > 0, 1 \leq i \leq T(\aGraph,\esDato{c})\}\\
\
& E_I = \bigcup_{v \in V_\aGraph} \{(v,\delta(D_\aGraph(v),\esDato{c}),u_{\esDato{c},i}) : v,u_{\esDato{c},i}\in V_I\}
\end{align*}

$I$ has two sets of nodes: the ones from $\aGraph$ and new $T(\aGraph,\esDato{c})$ nodes for every $\esDato{c}\in \Sigma_n$ such that $T(\aGraph,\esDato{c}) > 0$ (it follows from $\sum_{\esDato{c}\in\Sigma_n}T(\esDato{c})=N$ that there are $N$ nodes of the form $u_{\esDato{c},i}$). Observe that a matching in $\aProbabilisticDatabase$ is just an assignation of the nodes from $\aGraph$ to the desired data values.

It is easy to see that the assignation with the lowest cost corresponds to the matching with minimum weight.
\end{proof}

This model does not interact with the edges of the data-graph but simple interactions with the edges can be imposed by increasing the expressive power of the realization model: for example, we could consider local functions $\delta$ that take as input the edges labels outgoing the nodes being changed, instead of only the data values. As far as these interactions are independent of the data values, the modelling through weighted bipartite matching will remain valid.

\subsubsection{\nombreHittingSetSection}\label{hittingSetSection}

For this next restricted version, let us consider the alphabet $\Sigma_n$ of data values and, for each node $v$ in $\aGraph'$, suppose we are given a subset of data values $X_v \subseteq \Sigma_n$. We may consider that these subsets comprise all the valid data values for each node, trusting in some pre-processing or prior knowledge on the data.
For our problem, we assume that the data value of each $v$ in $\aGraph'$ is missing.
However, we do know that data-graphs $\aGraph$ such that $\aInverseRealizationModel(\aGraph')(\aGraph) >0$ are precisely those isomorphic to $\aGraph'$ for which the data values of each $v$ lies in $X_v$, following the intuitive idea that these data values represent the only valid assignments for each node.
Furthermore, those data-graphs in the cosupport of $\aGraph'$ that have the highest probability are those that have fewer distinct data values.

We define the problem \nombreHittingSet when the input is a PUDG $\aPUDG = (\aProbabilisticDatabase,\aRealizationModel,\aGraph')$ with an $\aRealizationModel$ as above, 
and a family $\{X_v\}_{v\in \aGraph'}$ such that $X_v \subseteq \Sigma_n$ and $|X_v|>0$.

\begin{theorem}
The problem {\nombreHittingSet} is NP-hard even if $|X_v|\leq 2$ for every $v \in \aGraph$.

\end{theorem}

\begin{proof}
It follows from the definition that finding a solution of maximum probability for this version of Node-Update EMDG problem is equivalent to finding a Hitting Set for the family $\{ X_v\}_{v\in \aGraph'}$, since we are looking for a minimum subset of data values $X \subset \Sigma_n$ such that $X \cap X_v \neq \emptyset$ for every $v \in \aGraph'$. Moreover, the decision version of this problem is NP-complete even if $|X_v| \leq 2$~\cite{garey1979computers}.
\end{proof}

Notice that we are considering a quite restricted version of the general Node-Update EMDG problem, still focusing the restrictions on cardinality-based constraints. And even in the very specific case in which we consider every subset $X_v$ of size smaller than 2, we obtain an NP-complete version of our original problem.

We could also consider the analogous version for Edge-Update EMDG. More precisely, we are given subset of edge-labels $X_{v,w} \subseteq \Sigma_e$ for each edge $(v, \esLabel{e}, w)$ in $E(\aGraph')$. As explained in the previous paragraphs, we consider that the subset $X_{v,w}$ comprises all those valid edge-label values for each edge $(v, \esLabel{e}, w)$, once again trusting in some pre-processing or prior knowledge.
Moreover, we assume that all the edge-labels are missing, and that the data-graphs $\aGraph$ for which $\aInverseRealizationModel(\aGraph')(\aGraph) >0$ are precisely those isomorphic to $\aGraph'$ such that the edge-labels $\esLabel{e}$ of each edge between $v$ and $w$ lie in $X_{v,w}$.
Furthermore, those data-graphs in the preimage of $\aGraph'$ that have the highest probability are those that have fewer distinct edge-labels. The hardness of this version follows analogously as in the Node-Update case.

\begin{theorem}
The problem {\nombreHittingSetEdge} is NP-hard even if $|X_{v,w}|\leq 2$ for every $(v, \esLabel{e}, w) \in E_\aGraph$.
\end{theorem}







\section{PUDG with \Gregxpath constraints}

In this section, we consider alternative ways to define a system of priorities by using $\Gregxpath$ expressions to add constraints.


\subsection{Adding constraints}\label{section:constraints}
While a particular EMDG $(\aProbabilisticDatabase, \aRealizationModel)$ can be used to encode information about our epistemic understanding, sometimes it makes sense to restrict our space of possible worlds in particular ways. Depending on the field of application, these restrictions could arise from obtaining a deeper theoretical understanding, as a way to explore multiple hypotheses or from practical limitations on the world-models that can be effectively analyzed.

The concept of hard and soft constraints from the area of database repairing \cite{barcelo2017data, tenCate:2012, arenas1999consistent} can be incorporated to our framework of data-graphs, for both problems of Data Cleaning and PQA.
 
\begin{definition}[Restrictions and consistency] 
Let $\aGraph$ be a data-graph and $\aRestrictionSet = \aRestrictionSetPaths \union \aRestrictionSetNodes$ a set of \defstyle{restrictions} (also called \defstyle{constraints}), where $\aRestrictionSetPaths$ consists of path expressions and $\aRestrictionSetNodes$ of node expressions. 
We say that $\aGraph$ \defstyle{violates} $\aNodeExpression$, if $\aNode \not\in \semantics{\aNodeExpression}_\aGraph$ for some $ \aNode \in V_\aGraph$. Otherwise, we say that $\aGraph$ \defstyle{satisfies} $\aNodeExpression$, and we denote it by $\aGraph \models \aNodeExpression$.
Similarly, we say that $\aGraph$ \defstyle{violates} $\aPath$ if $(\aNode,\aNodeb) \not\in \semantics{\aPath}_\aGraph$ for some $\aNode,\aNodeb \in V_\aGraph$, and otherwise we say that $\aGraph$ \defstyle{satisfies} $\aPath$ ($\aGraph \models \aPath$).

If $\aGraph$ does not violate $\aFormulaNodeOrPath$ for every $\aFormulaNodeOrPath \in \aRestrictionSet$, we say that $\aGraph$ is \defstyle{consistent} w.r.t.\ $\aRestrictionSet$, and we note this as $\aGraph \models \aRestrictionSet$.

We say that $\aGraph,\aNode$ \defstyle{violates} $\aNodeExpression$, if $\aNode \not\in \semantics{\aNodeExpression}_\aGraph$. 
Similarly, we say that $\aGraph, \aNode, \aNodeb$ \defstyle{violates} $\aPath$ if $(\aNode,\aNodeb) \not\in \semantics{\aPath}_\aGraph$.
\end{definition}

This notion of consistency is useful for dividing the set of data-graphs into two natural groups: those who satisfy the constraints and those who do not. However, this stark division can be generalized.

Given a set of constraints $\aRestrictionSet$, we can assign a weight $0 \le w_\aFormulaNodeOrPath < 1$ to each $\aFormulaNodeOrPath \in \aRestrictionSet$ by considering a function $\aRestrictionSetWithWeights: \aRestrictionSet \to \left[0,1 \right)$.

A value of $0$ can be thought of as a \defstyle{hard (exclusionary) constraint}: 
we want our clean models to never satisfy the formula; these constraints can be thought of as a type of denial constraints~\cite{Bertossi_2011_book, Fan_2012}. On the other hand, a value between $0$ and $1$ is a \defstyle{soft constraint}: the models that satisfy the formula are penalized in their probability, but not absolutely. 

A weight function $\aRestrictionSetWithWeights$ defines a preference relation among constraints. Given $(\aProbabilisticDatabase, \aRealizationModel)$ an EMDG,
we would like to
lift this preference relation to a preference relation among the
data-graphs $\aGraph$ in $\aProbabilisticDatabase$, combining
its associated probability with
the weights from
each constraint in
$\aRestrictionSet$
that is satisfied by $\aGraph$. We formalize this as follows: 

\[
\aRestrictionSetWithWeights(\aGraph) \eqdef \aProbabilisticDatabase(\aGraph) \times \prod_{\aFormulaNodeOrPath \in \aRestrictionSet} \aRestrictionSetEvaluated(\aGraph, \aFormulaNodeOrPath)
\]

\noindent where $\aRestrictionSetEvaluated(\aGraph, \aFormulaNodeOrPath)$ is defined as: 

\[ \aRestrictionSetEvaluated(\aGraph, \aFormulaNodeOrPath) \eqdef \begin{cases} 
     \aRestrictionSetWithWeights(\aFormulaNodeOrPath) & \mbox { if $\aGraph \models \aFormulaNodeOrPath$}\\
     1  & \mbox{ otherwise}
   \end{cases}
\]

Given $(\aProbabilisticDatabase, \aRealizationModel)$ an EMDG, we can now
define the new EMDG $(\aProbabilisticDatabase_\aRestrictionSetWithWeights, \aRealizationModel)$
which is an EMDG where the weights given by $\aRestrictionSetWithWeights$ have been incorporated and normalized from the original values of $\aProbabilisticDatabase$.

\[
\aProbabilisticDatabase_\aRestrictionSetWithWeights(\aGraph) 
\eqdef \frac{\aRestrictionSetWithWeights(\aGraph)}{\displaystyle\sum_{\aGraphb \in \aProbabilisticDatabase} \aRestrictionSetWithWeights(\aGraphb)}  
\]

Notice that incorporating $\aRestrictionSetWithWeights$ does not introduce modifications in $\aRealizationModel$, but it does change $\aInverseRealizationModel$ (see~(\ref{eq:Bayes}) for the formal definition of $\aInverseRealizationModel$). 

\subsubsection{Complexity considerations}
Given an EMDG $(\aProbabilisticDatabase, \aRealizationModel)$ and a weighted restriction set $\aRestrictionSetWithWeights$, how much harder are the problems of Data cleaning and PQA over $\aProbabilisticDatabase_\aRestrictionSetWithWeights$?
Observe that, since $\aProbabilisticDatabase$ could contain an infinite number of data-graphs of positive probability, computing $\aProbabilisticDatabase_\aRestrictionSetWithWeights$ explicitly might not be possible. 

\begin{example}
Consider a simple case where $\aRestrictionSet$ consists of a single node or a single path expression, this is, $\aRestrictionSet = \{\aFormulaNodeOrPath\}$. Furthermore, assume that we know the total probability mass $p_\aFormulaNodeOrPath$ of data-graphs in $\aProbabilisticDatabase$ where $\aFormulaNodeOrPath$ holds. Then if follows from a simple calculation that: 
$$\aProbabilisticDatabase_\aRestrictionSetWithWeights(\aGraph) = \frac{\aProbabilisticDatabase(\aGraph) \times \aRestrictionSetEvaluated(\aGraph, \aFormulaNodeOrPath)}{p_\aFormulaNodeOrPath \times \aRestrictionSetWithWeights(\aFormulaNodeOrPath) + (1 - p_\aFormulaNodeOrPath) }$$

Indeed, this is a consequence of the definition of $\aProbabilisticDatabase_\aRestrictionSetWithWeights(\aGraph)$ and rewriting the denominator as follows:

\begin{align*} 
\sum_{\aGraphb \in \aProbabilisticDatabase} \aRestrictionSetWithWeights(\aGraphb) &= \sum_{\aGraphb \in \aProbabilisticDatabase\, \mid\, \aGraphb \models \aFormulaNodeOrPath} \aRestrictionSetWithWeights(\aGraphb) + \sum_{\aGraphb \in \aProbabilisticDatabase\, \mid\, \aGraphb \not\models \aFormulaNodeOrPath} \aRestrictionSetWithWeights(\aGraphb)\\
& = \sum_{\aGraphb \in \aProbabilisticDatabase\, \mid\, \aGraphb \models \aFormulaNodeOrPath} (\aProbabilisticDatabase(\aGraphb) \times \aRestrictionSetEvaluated(\aGraphb, \aFormulaNodeOrPath)) + \sum_{\aGraphb \in \aProbabilisticDatabase\, \mid\, \aGraphb \not\models \aFormulaNodeOrPath} (\aProbabilisticDatabase(\aGraphb) \times\aRestrictionSetEvaluated(\aGraphb, \aFormulaNodeOrPath)) \\
& = \sum_{\aGraphb \in \aProbabilisticDatabase\, \mid\, \aGraphb \models \aFormulaNodeOrPath} (\aProbabilisticDatabase(\aGraphb) \times \aRestrictionSetWithWeights(\aFormulaNodeOrPath)) + \sum_{\aGraphb \in \aProbabilisticDatabase\, \mid\, \aGraphb \not\models \aFormulaNodeOrPath} (\aProbabilisticDatabase(\aGraphb) \times 1) \\
& =  \aRestrictionSetWithWeights(\aFormulaNodeOrPath) \times \sum_{\aGraphb \in \aProbabilisticDatabase\, \mid\, \aGraphb \models \aFormulaNodeOrPath} \aProbabilisticDatabase(\aGraphb)+ \sum_{\aGraphb \in \aProbabilisticDatabase\, \mid\, \aGraphb \not\models \aFormulaNodeOrPath} \aProbabilisticDatabase(\aGraphb) \\
& =  \aRestrictionSetWithWeights(\aFormulaNodeOrPath) \times p_\aFormulaNodeOrPath + (1-p_\aFormulaNodeOrPath)
\end{align*}
\end{example}
Notice that the same idea applies for more complex sets of constraints, as long as we know the probability masses in $\aProbabilisticDatabase$ for each combination of Boolean values of the formulas in $\aRestrictionSet$ (e.g., if $\aRestrictionSet = \{\aNodeExpression, \aNodeExpressionb\}$, we ask for $p_{\aNodeExpression, \aNodeExpressionb}, p_{\lnot\aNodeExpression, \aNodeExpressionb}, p_{\aNodeExpression, \lnot\aNodeExpressionb}, p_{\lnot\aNodeExpression, \lnot\aNodeExpressionb}$, respectively meaning the total probability of data-graphs satisfying the two formulas in the sub-index).



\subsection{Expression-based Constraints in Node-Update PUDGs}\label{Section:ExpressionBasedConstraints}

We now explore some particular instances of the Data cleaning problem in Node-Update PUDGs. As already mentioned, we focus on this kind of PUDGs given that data values are one of the main features of our data-graph model.

An alternative way to define a system of priorities with the distribution $\aProbabilisticDatabase$ is by using a set of \textit{path expressions} or \textit{integrity constraints} that can be evaluated on a data-graph. For example, we consider an expression $\aFormulaNodeOrPath$ from a language such as $\Gregxpath$ (see Section~\ref{Section:Definitions} for the basic definitions) and define $\aProbabilisticDatabase_\aFormulaNodeOrPath$ in some way such that $\aProbabilisticDatabase_\aFormulaNodeOrPath(\aGraph) < \aProbabilisticDatabase_\aFormulaNodeOrPath(\aGraph')$ whenever $\aGraph \not\models \aFormulaNodeOrPath$ and $\aGraph' \models \aFormulaNodeOrPath$. Then, we prioritize those data-graphs that satisfy some topological structure that we are expecting to see in the clean database, which is expressible through some language $\mathcal{L}$.

We may also `evaluate' the expression $\aFormulaNodeOrPath$ only in a specific node (sometimes called the origin $o$) or in a pair of nodes, as it was first proposed in~\cite{abiteboul1999regular} regarding path constraints. This restriction might turn reasoning problems easier (see for example~\cite{barcelo2017data}). From now on, we denote the universe of data-graphs that have a distinguished node from where node expressions can be evaluated as $U_o$, and given $G \in U_o$ we denote this distinguished node as $o_\aGraph$. 

Nonetheless, we point out some observations regarding the relationship between the semantics when evaluating an expression in all nodes and when only considering an origin $o$.

\begin{observation}\label{observation:equivalence}
Relation between global and origin semantics.

\begin{enumerate}
    \item There exists a function $f:U \times \Gposregxpath \to U_o' \times \Gposregxpath$ computable in polynomial time such that for any data-graph $\aGraph$ and $\Gposregxpath$ node expression $\aNodeExpression$ is true that $\aGraph \models \aNodeExpression$ if and only if $\aGraph', o_{\aGraph'} \models \aNodeExpression'$, where $f(\aGraph, \aNodeExpression) = (\aGraph', \aNodeExpression')$.
    
    \item There exists a function $g:U_o' \times \Gposregxpath \to U \times \Gposregxpath$ computable in polynomial time such that for any $\aGraph \in U_o$ and $\Gposregxpath$ expression $\aNodeExpression'$ is true that $\aGraph', o_{\aGraph'} \models \aNodeExpression'$ if and only if $\aGraph \models \aNodeExpression$, where $g(\aGraph', \aNodeExpression') = (\aGraph, \aNodeExpression)$
\end{enumerate}

\end{observation}

\begin{proof}
For (1), we can construct $\aGraph'$ in the following way: consider an arbitrary order $p_1,\ldots,p_n$ of the nodes of $\aGraph$ and a new node $p_0$ with any data value. Now, add $p_0$ to $\aGraph$ with a loop on a fresh edge label $loop$ and define a Hamiltonian cycle over $\aGraph$ by using another new edge label $\epsilon$ and adding edges $(p_i, \epsilon, p_{i+1})$ for every $0\leq i \leq n$ where $p_{n+1} = p_0$. Observe that $\aGraph \models \aFormulaNodeOrPath$ if and only if $\aGraph',p_0 \models \comparacionCaminos{(\down_{\epsilon} [f(\aFormulaNodeOrPath)])^* \down_\epsilon \down_{loop}}$, where $f(\aNodeExpression)$ denotes the expression $\aNodeExpression$ after replacing any $\_$ subformula by $\bigcup\limits_{l: \Sigma_e} \down_l$.\footnote{This is necessary to prevent the formula $\aNodeExpression$ from using the edges $\down_e$ in $\aGraph'$.}

For (2), we build $\aGraph$ by adding to every node of $\aGraph'$ different from $o$ a loop with edge label $loop$. It follows that $\aGraph', o \models \aNodeExpression'$ if and only if $\aGraph \models \down_{loop} \cup [f(\aNodeExpression')]$ where $f$ is the same function as in (1) that replaces the use of the wildcard $\_$.\footnote{In both directions, the sets $\Sigma_n$ and $\Sigma_e$ might need to be extended in order to find a fresh data value or edge label}
\end{proof}

\begin{remark}
When considering path expressions instead of node expressions, there is a result analogous to Observation~\ref{observation:equivalence} that relates global semantics with bi-pointed semantics. 
Indeed, it is easy to see that $\aGraph, \aNode, \aNodeb \models \aPathExpression$ if and only if $\aGraph' \models \aPathExpression'$, where $\aGraph'$ is constructed by adding fresh edge labels $\esLabel{i}, \esLabel{f}$ (expanding $\Sigma_e$), and adding an edge $\esLabel{i}$ from all the nodes to $\aNode$, and an edge $\esLabel{f}$ from $\aNodeb$ to all the nodes. Then, we consider $\aPathExpression' = \downarrow_\esLabel{i} \aPathExpression \downarrow_{\esLabel{f}}$. 

For the other direction, we have that $\aGraph \models \aPathExpression$ if and only if $\aGraph', \aNode, \aNodeb \models \aPathExpression'$, where $\aGraph$ is modified into $\aGraph'$ by adding two fresh nodes $\aNode$ and $\aNodeb$, adding edges with a fresh label $\esLabel{i}$ from $\aNode$ towards all the nodes in $\aGraph$, and an edge with fresh label $\esLabel{f}$ from all the nodes in $\aGraph$ to $\aNodeb$. Then, we consider 
$\aPathExpression' = \complementoCamino{\downarrow_{\esLabel{i}} \complementoCamino{\aPathExpression} \downarrow_\esLabel{f}}$.

Notice, nonetheless, that for these results we had to use negation to build the new expression. Then, if we want to remain in the positive fragment, we need to consider only node expressions.

\end{remark}

Observation~\ref{observation:equivalence} implies that we may interchange the semantics considered for our reasoning problems, as long as the new node expression belongs to the same set of expressions $\mathcal{L}$ and we consider Node-Update PUDGs. Notice that, to go from the global case to the origin one --when considering node expressions-- we had to use the Kleene operator over a path expression involving nested sub-expressions and a data test. Any language that is useful for defining constraints in data-graphs must have data-tests, and therefore a reasonable set of expressions to study over the origin semantics whose results do not follow naturally from the global scenario could be the subset of \Gposregxpath formulas that only use the Kleene star in restricted cases. In fact, we will show that considering a restricted version of this operator allows for polynomial algorithms to solve the \decisionProblem{Data cleaning} problem in the context of Node-Update PUDGs.

Naturally, the difficulty of finding the clean database $H$ that maximizes $\aInverseRealizationModel_{\aProbabilisticDatabase, \aRealizationModel}(H, \aGraph)$ will depend on the complexity of the language $\mathcal{L}$ used to define the expressions. For example, if $\Gposregxpath \subseteq \mathcal{L}$ then the problem is intractable:

\begin{theorem}\label{theorem:expressionIntractable}

There are fixed finite sets $\Sigma_n$ and $\Sigma_e$, and a fixed $\Gposregxpath$ formula $\aFormulaNodeOrPath$ such that given a data-graph $\aGraph$, the problem of deciding if there exists a data-graph $H$ isomorphic to $\aGraph$ up to its data values such that $H \models \mu$ is \textsc{NP-complete}.
\end{theorem}

\begin{proof}
The problem is in $\textsc{NP}$: we can guess a data-graph $H$ and check that $H \models \aFormulaNodeOrPath$ and $\aGraph \equiv H$. Now, we reduce our problem to \textsc{3SAT} to prove \textit{hardness}.\footnote{Here we use the assumption that data values are $O(1)$ to find a polynomial-sized witness. Nevertheless, this theorem remains true in a more general scenario where we allow the size of data values to grow logarithmically because, if there is a repair of $\aGraph$ that works as a witness, then there is another that uses the data values in $\Sigma_n^\aFormulaNodeOrPath$ and at most $|V_\aGraph|$ more, where these other data values not in $\Sigma_n^\aFormulaNodeOrPath$ can be chosen arbitrarily (see Lemma~\ref{lemma:boundedNodes}).}

Given a 3SAT boolean formula $\phi$ of $n$ variables $x_1\ldots x_n$ and $m$ clauses $c_1 \ldots c_m$ we will build a data-graph $\aGraph$ and a $\Gposregxpath$ node expression $\aNodeExpression$ such that $\phi$ is satisfiable if and only if there exists a data-graph $\aGraph'$ such that $\aGraph' \equiv \aGraph$ and $\aGraph' \models \aNodeExpression$.

Consider $\Sigma_n = \{\esDato{var}, \esDato{clause}, \bot, \top \}$ and $\Sigma_e = \{ \down_-,\down_+,\down_{\esLabel{notClause}} \}$ and define $\aGraph=(V,L,D)$ as:

\begin{align*}
& V  = \{x_i:1\leq i \leq n\} \cup \{c_j : 1\leq j \leq m\} \\
& L(x_i,c_j)  = \{\down_+\} \iff x_i \mbox{ is a literal from } c_j \\
& L(x_i,c_j)  = \{\down_-\} \iff \lnot x_i \mbox{ is a literal from } c_j \\
& L(x_i,x_i)  = \{\down_{\esLabel{notClause}}\} \\
& L(z_1,z_2)  = \emptyset \mbox{ for any other case} \\
& D(c_i)  = \esDato{clause} \\
& D(x_i)  = \esDato{var} 
\end{align*}

This data-graph encodes the formula $\phi$ through the edges $\down_+$ and $\down_-$. Observe that, given a valuation of the variables of $\phi$, we can set the data values of the node variables of $\aGraph$ accordingly to $\bot$ or $\top$ to represent that valuation. Moreover, we can build a formula that only captures those data-graphs representing an assignment that satisfies $\phi$ with the node expression 
$$
\aNodeExpression_1 = \comparacionCaminos{\down_+^-[\top]} \vee \comparacionCaminos{\down_-^-[\bot]} \vee [\esDatoDistinto{clause}].
$$

This formula forces every node with data value \esDato{clause} to have a path through edges $\down_{-}^-$ and $\down_{+}^-$ to a value $\bot$ or $\top$, respectively.

We use the $\esLabel{notClause}$ edges to forbid any isomorphic data-graph $H$ from changing the data value of the clause nodes in the following way:
$$
\aNodeExpression_2 = \comparacionCaminos{\down_{\esLabel{notClause}}} \vee [\esDatoIgual{clause}]
$$

That is, we only want to consider data-graphs in which the data values that change are those from the variable nodes, which should remain the same, or change to either $\top$ or $\bot$.

Finally, we set $\aNodeExpression = \aNodeExpression_1 \wedge \aNodeExpression_2$. It follows that $\phi$ is satisfiable if and only if there exists a data-graph $\aGraph'$ such that $\aGraph' \equiv \aGraph$ and $\aGraph \models \aNodeExpression$. Observe that given $\phi$ we can build $\aGraph$ in linear time, and therefore we reduced \textsc{3SAT} to our problem.
\end{proof}

Given a data-graph $\aGraph$, the problem of deciding the existence of an isomorphism $\aGraph'$ of $\aGraph$ that satisfies a certain formula $\aFormulaNodeOrPath$ may be interpreted as a particular relaxation of the data cleaning problem, in which the only data-graphs with positive weights are those that satisfy $\aFormulaNodeOrPath$ and $\aRealizationModel(\aGraph')(\aGraph) = 1$ for every data-graph $H$. Moreover, these data-graphs also have the same weight, and therefore are ``equally likely'' to be the clean data-graph. To understand this ``lower bound''' to the data cleaning problem in the context of expression-based constraints we study the following problem:

\begin{center}
\fbox{\begin{minipage}{25em}
  \textsc{Problem}: \decisionProblem{Isomorphic-repair}

\textsc{Input}: A data-graph $\aGraph$ and a formula $\aFormulaNodeOrPath$ from some set of expressions $\mathcal{L}$

\textsc{Output}: Decide whether there exists a data-graph $H$ such that $H \equiv \aGraph$ and $H \models \aFormulaNodeOrPath$.
\end{minipage}}
\end{center}

This problem is similar to the \textit{repair computing} problem, deeply studied in the database theory community (see~\cite{Bertossi_2011_book}).

We already showed that the problem is intractable if $\Gposregxpath \subseteq \mathcal{L}$. Furthermore, in Theorem~\ref{theorem:expressionIntractable} we considered the node expression \textit{fixed}. This is a common simplification, usually referred to as the \textit{data complexity} of the problem \cite{Vardi82thecomplexity}. Also, notice that in Observation~\ref{observation:equivalence} the formula constructed in the first remark has fixed size if $\aFormulaNodeOrPath$ is fixed. Therefore, we conclude that: 

\begin{corollary}

The problem \decisionProblem{Isomorphic repair} is \textsc{NP-complete}, even if the node expression $\aFormulaNodeOrPath \in \Gposregxpath$ is fixed and the expression is evaluated only in an origin $o$ from $V_\aGraph$.

\end{corollary}

As previously remarked, one way to weaken the language $\mathcal{L}$ in order to avoid concluding intractability as a consequence of Theorem~\ref{theorem:expressionIntractable} and Observation~\ref{observation:equivalence}, would be to limit the cases in which we allow the Kleene star. If we completely remove it, then the problem remains hard, but only under very particular assumptions:

\begin{theorem}
The \decisionProblem{Isomorphic-repair} is \textsc{NP-complete}, even if we only evaluate the expression at an origin $o$, as long as we allow:

\begin{itemize}
    \item The set of data values to be infinite (i.e. $|\Sigma_n| > \infty$)
    
    \item The expression not to be fixed, and the use of the $\comparacionCaminos{\cdot}$ operator and equality data tests ($\esDatoIgual{\esDato{c}})$.
\end{itemize}
\end{theorem}

\begin{proof}
Membership in \textsc{NP} follows from  Theorem~\ref{theorem:expressionIntractable}. We prove the hardness by reducing \decisionProblem{HAMILTONIAN PATH} from a distinguished vertex $v$ to our particular case of \decisionProblem{Isomorphic repair}.

Let $\Sigma_n = \N_0$ and $\Sigma_e = \{\esLabel{down}\}$. Given a directed graph $\aGraph = (V,E)$ we construct the data-graph $\aGraph' = (V',L',D')$ with origin $v$ as follows:
\begin{align*}
& V' = V \\
& L(z, w) = \{\esLabel{down}\} \iff (z, w) \in E \text{    }\forall z,w \in V \\
& D(z) = 0 \text{    }\forall z \in V.
\end{align*}

We define the node expression:
$$
\aNodeExpression = \comparacionCaminos{[\esDatoIgual{1}] \down_\esLabel{down} [\esDatoIgual{2}] \down_\esLabel{down} [\esDatoIgual{3}]
\ldots
\down_\esLabel{down} [\esDatoIgual{n}]}
$$

It is easy to see that we can find a new data-graph $H$ equal to $\aGraph$ up to its data values such that $H, v \models \aNodeExpression$ if and only if $\aGraph$ has a Hamiltonian path starting at $v$.
\end{proof}

The formula used for this proof is extremely simple, however it does depend on the input graph and it strongly relies on $\Sigma_n$ being infinite. 

Now, let us consider a weaker family of expressions, namely the ones defined by the grammar:

\begin{center}
    $\aPath, \aPathb$ = $\epsilon$ $|$ $\aLabel$ $|$ $\aLabel^{-}$ $|$ $\expNodoEnCamino{\aFormula}$ $|$ $\aPath$ . $\aPathb$ $|$ $\aPath \pathUnion \aPathb$ $|$ $\aPath \pathIntersection \aPathb$ $|$ $\aLabel^{*}$ $|$ $\aLabel^{-*}$ $|$ $\aPath^{n,m}$
\end{center}  

\begin{center}
    $\aFormula, \aFormulab$ = $\aFormula \wedge \aFormulab$ $|$ $\comparacionCaminos{\aPath}$ $|$ $\esDatoIgual{\aData}$ $|$ $\esDatoDistinto{\aData}$ $|$ $\comparacionCaminos{\aPath = \aPathb}$ $|$ $\comparacionCaminos{\aPath \neq \aPathb}$ $|$
    $\aFormula \vee \aFormulab$
\end{center}

This is a subset of \Gregxpath called \Gposcoreregxpath~\cite{libkin2016querying}.

Observe that we removed the operation that allowed us to use the Kleene star over any path expression. Now, we can only use it on atomic edge labels, or inverse of edge labels. In particular, this implies that we cannot ask for the transitive closure of path expressions that interact with data.

We make the following observation:

\begin{observation}\label{observation:deleteKleene}
 There is a function $f:U \times \Gposcoreregxpath \to U \times \Gposregxpath$ computable in polynomial time such that $\semantics{\aGraph}_\aFormulaNodeOrPath = \semantics{\aGraph'}_{\aFormulaNodeOrPath'}$, $\aGraph'$ has polynomial size on $\aGraph$, $\aFormulaNodeOrPath'$ has polynomial size on $\aFormulaNodeOrPath$ and $\aFormulaNodeOrPath'$ does not use the $*$ operator.
\end{observation}

\begin{proof}

For every edge label $l \in \Sigma_e$, we create a new label $l^*$ and add the edges $(v,l^*,w)$ to $\aGraph$ for each pair of nodes $v,w \in V_\aGraph$ such that $(v,w) \in \semantics{\aGraph}_{\down_{l}^*}$. This is our desired $\aGraph'$. To define $\aFormulaNodeOrPath'$, we only need to change any subformula $\down_l^*$ to $\down_{l^*}$ in $\aFormulaNodeOrPath$.
The same can be done for the subformulas of the form $\down_l^{-*}$.
\end{proof}

Actually, this procedure may be utilized even when the Kleene star is applied to arbitrary path expressions $\aPath$, by defining a new edge label $\down_{\aPath^{\star}}$ and joining every pair of nodes $v,w$ such that $(v,w) \in \semantics{\aPath^*}$. However, the difference is that if we only use the $\ast$ operator in atomic expressions, then the fact that $(v,w) \in \semantics{\aPath^*}$ is independent of the data values of the data-graph. Due to this fact, we deduce that:\footnote{It is enough to restrict the $\ast$ operator to be used only in subexpressions $\aFormulaNodeOrPath$ that do not contain any data operation (i.e. a datatest $[\esDatoIgual{c}]$ or $[\esDatoDistinto{c}]$ or a data comparison like $\comparacionCaminos{\aPath = \aPathb}$ or $\comparacionCaminos{\aPath \neq \aPathb}$).}

\begin{lemma}\label{lemma:boundedNodes}
 Given a \Gposregxpath formula $\aFormulaNodeOrPath$ that does not use the Kleene star, there is a constant $c_\aFormulaNodeOrPath$ such that if $(v,w) \in \semantics{\aFormulaNodeOrPath}_\aGraph$ for some data-graph $\aGraph$, then there is a subset of $\aGraph$ with at most $c_\aFormulaNodeOrPath$ nodes $V_\aFormulaNodeOrPath$ such that $(v,w) \in \semantics{\aFormulaNodeOrPath}_{\aGraph_{V_\aFormulaNodeOrPath}}$, where $\aGraph_X$ is the subgraph of $\aGraph$ induced by the nodes $X$.
\end{lemma}

\begin{proof}
We prove this by induction on the formula's structure, in order to obtain a tight bound. 

The base cases are trivial. This is, $c_\epsilon = 1$, $c_{\down_l} = 2$, $c_{[\esDatoIgual{d}]} = 1$, and so on.

For the inductive case, it is easy to see that $c_{\aPath . \aPathb} = c_\aPath + c_\aPathb - 1$, $c_{\aPath \cup \aPathb} = \min(c_\aPath, c_\aPathb)$, $c_{\aPath \cap \aPathb} = c_\aPath + c_\aPathb - 1$, and so on. The most interesting case is $\aPath^{n,m}$, where we have that $c_{\aPath^{n,m}} = m \times c_{\aPath}$.
\end{proof}

Observe that $c_\aFormulaNodeOrPath$ might be exponential with respect to $|\aFormulaNodeOrPath|$, since in $\aPath^{n,m}$ the size of $m$ is $\log(m)$.

Notice that, given a $\Gposcoreregxpath$ formula without the Kleene star, now we can obtain a bound for the number of nodes that interact with the formula by evaluating it at an origin. Hence, to decide whether a data-graph $\aGraph$ satisfies a given formula $\aFormulaNodeOrPath \in \Gposcoreregxpath$ from an origin $o$, it suffices to evaluate it in every data-graph $\aGraph' \subseteq \aGraph$ where $|V_{\aGraph'}| < |c_\aFormulaNodeOrPath|$ and $o \in V_{\aGraph'}$. It follows that, if the formula is satisfied by any $\aGraph'$ then $\aGraph$ will also satisfy it (due to the monotonicity of \Gposregxpath, see Section~\ref{Section:Definitions}). 

\noindent We will show a polynomial restriction of the \decisionProblem{isomorphic-repair} problem, but first we need the following lemma:

\begin{lemma}\label{lema:constraints1}
Let $A$ be a subset of $n$ distinct data-values that do not appear in $\eta$. If a data-graph of $n$ nodes satisfies $\eta$, then there is another data-graph that satisfies $\eta$ and only uses data-values from $\Sigma_n^\eta \cup A$.  
\end{lemma}

\begin{proof}
    The main idea behind this lemma is that the ``identity'' of those data values that are not in $\Sigma_n^\eta$ is not important, since they only satisfy expressions of the form $[c]^\neq$, $\langle \aPath \neq \aPathb \rangle$ or $\langle \aPath = \aPathb \rangle$. Therefore, if we modify them still preserving the equalities between them and the fact that they are not mentioned in $\eta$, the satisfiability of $\eta$ is not affected. The details are in the Appendix~\ref{section:Appendix}.
\end{proof}

Based on Lemmas~\ref{lemma:boundedNodes} and \ref{lema:constraints1}, we obtain the following result:

\begin{theorem}\label{theorem:repairTratable}

Let $\aNodeExpression \in \Gposcoreregxpath$ be a fixed node expression that is evaluated in an origin, then the \decisionProblem{Isomorphic-repair} problem can be solved in polynomial time.

\end{theorem}

\begin{proof}

Given the data-graph $\aGraph$ and the expression $\aNodeExpression$, we start by removing the Kleene operator from $\aNodeExpression$ as in Observation~\ref{observation:deleteKleene}. Let $\aNodeExpression'$ be this new formula, alongside our new data-graph $\aGraph'$. To decide whether there exists a data-graph $H$ such that $H \equiv \aGraph'$ and $H, o \models \aNodeExpression$, it is enough to find a sub data-graph $H'$ induced by some nodes of $\aGraph'$ containing $o$ with size less or equal to $c_{\aNodeExpression'}$ such that $H', o \models \aNodeExpression'$.

Notice that we only need to consider $\bigO{n^{c_{\aNodeExpression'}}}$ such sub data-graphs. For each of these data-graphs, we should consider every possible assignment of data values to the nodes. Apart from those data values in $\Sigma_n^\aNodeExpression$, we might need some extra data values, since it could be the case that we need to satisfy some subformula of the form $\bigwedge_{\esDato{c} \in \Sigma_n^\aNodeExpression} \esDato{c}^\neq$; or also an inequality of data values such as $\comparacionCaminos{\aPath_1 \neq \aPath_2}$.

By Lemma~\ref{lema:constraints1}, it suffices to choose a set $A$ of $c_\aNodeExpression$ arbitrary data values not mentioned in $\aNodeExpression$. Then, for every sub data-graph $H'$ of at most $c_\aNodeExpression$ nodes we need to consider each data value assignment of its nodes that only uses data values from $\aNodeExpression \cup A$. We can bound the number of such assignments by $(|\aNodeExpression| + c_\aNodeExpression)^{c_\aNodeExpression}$. Since $\aNodeExpression$ is fixed, this bound is $\bigO{1}$. Finally, noticing that checking if a certain data-graph satisfies a given node expression at an origin $o$ can be done in polynomial time on the data-graph and formula size, we conclude that the algorithm that checks every possible data value assignment for each possible sub data-graph $H'$ of at most $c_\aNodeExpression$ nodes of $\aGraph'$ runs in poly($|\aGraph|$).
\end{proof}

Now we focus our attention on the \decisionProblem{data cleaning} problem for the particular relaxation in which the \decisionProblem{Isomorphic-repair} is tractable. First, we state a lemma that will be useful in the proof of the main result, which is then presented in Theorem~\ref{teo:datacleaning_origin_poly}.

\begin{lemma}\label{lemma:optimumDatavalues}
    Let $\aGraph=(V_\aGraph, L_\aGraph, D_\aGraph)$ be a data-graph of $n$ nodes with an origin $o$, let $\aNodeExpression$ be a node expression from $\Gposcoreregxpath$ that does not contain any formula $\langle \aPath = \aPathb \rangle$ as subexpression, and let $\delta: \Sigma_n \times \Sigma_n \to \N_0$ be a function such that, for each fixed $\esDato{c} \in \Sigma_n$, we can efficiently enumerate the set $\esDato{d}_1, \esDato{d}_2, \ldots $ where $\delta(\esDato{d}_i, \esDato{c}) \leq \delta(\esDato{d}_{i+1}, \esDato{c})$, and every data value $\esDato{d}$ belongs to the enumeration.
    
    If there is a data-graph $H=(V_\aGraph, L_\aGraph, D_H)$ that satisfies $\aNodeExpression$ and minimizes $\sum\limits_{v \in V_\aGraph} \delta(D_H(v), D_\aGraph(v))$, then there is another data-graph $H' = (V_\aGraph, L_\aGraph, D_{H'})$ 
    that fulfills the same properties as $H$ and such that, for each data value $\esDato{c}$ in $\aGraph$ that was modified in $H'$ to some data value $\esDato{d}$ that is not mentioned in $\aNodeExpression$, satisfies that $\esDato{d}=\esDato{d}_j$ for $j \leq n + |\Sigma_n^\aNodeExpression|$ where $\esDato{d}_i$ is the $i$th data value in the efficient enumeration obtained for $\esDato{c}$ with respect to $\delta$. 
\end{lemma}

\begin{proof}
    The main idea is that those nodes that are mapped to data values outside the set $\Sigma_n^\aNodeExpression$ can be remapped to others (and preferably those that appear first in the enumeration $\esDato{d}_1, \esDato{d}_2, \ldots$), as long as this does not provoke that a subexpression of the form $\langle \aPath \neq \aPathb \rangle$ is not satisfied. We can guarantee that this will not happen by choosing a data value between the first $n + |\Sigma_n^\aNodeExpression|$, since there will always be some data value here that is not used in the data-graph. The details are in the Appendix~\ref{section:Appendix}.
\end{proof}

\begin{theorem}\label{teo:datacleaning_origin_poly}
The \decisionProblem{Data cleaning} problem can be solved in polynomial time when considering a PUDG $(\aProbabilisticDatabase, \aRealizationModel, \aGraph)$ and an origin $o$ such that:

\begin{itemize}
    \item $\aProbabilisticDatabase$ depends on a fixed node expression $\aNodeExpression \in \Gposcoreregxpath$ that does not contain a data comparison with equality (i.e., the formula $\langle \aPath = \aPathb \rangle$) as a subexpression, and such that $\aProbabilisticDatabase(H) > 0$ if and only if $H, o \models \aNodeExpression$ and $\aProbabilisticDatabase(H) = \aProbabilisticDatabase(H')$ if $H,o \models \aNodeExpression$ and $H',o \models \aNodeExpression$.\footnote{At least one data-graph $H$ must exist for this distribution to be well defined.}
    
    \item $\aRealizationModel$ depends on a local function $\delta$, as those considered in Section~\ref{section:localAndGlobal}, such that, given a data value $\esDato{c}$ we can efficiently enumerate all data values $\esDato{d}_1,\esDato{d}_2,\esDato{d}_3...$ so that $\delta(\esDato{d}_i,\esDato{c}) \leq \delta(\esDato{d}_{i+1},\esDato{c})$.
    
\end{itemize}
\end{theorem}

\begin{proof}
We will assume that $\aNodeExpression$ does not use the Kleene star operator, otherwise we apply first the transformation of Observation~\ref{observation:deleteKleene}.

We need to find a data-graph $H$ such that $H,o \models \aNodeExpression$, and $H$ minimizes the cost of the data modifications between $H$ and $\aGraph$. Naturally, we only need to modify the data values of the sub data-graph that satisfies $\aNodeExpression$, since leaving the remaining data values as they are is the most inexpensive strategy\footnote{Notice that this property of local functions as defined in Section~\ref{section:localAndGlobal} is not really needed for this algorithm.}.
To do this, we compute the cost of the transition for every sub data-graph and every assignment of data values in the algorithm given in Theorem~\ref{theorem:repairTratable}. Notice that, if any set of nodes is mapped to a set of data values that are not mentioned in $\aNodeExpression$, then we use the property imposed for $\delta$ to compute the most inexpensive data values from $\Sigma_n \setminus \Sigma_n^\aNodeExpression$. 

It follows from Lemma~\ref{lemma:boundedNodes} that there is a data-graph that coincides with $\aGraph=(V_G, L_G, D_G)$ except for its data values, and satisfies $\aNodeExpression$ if and only if there is a subset $V_\aNodeExpression \subseteq V(\aGraph)$ such that $|V_\aNodeExpression|= c_\aNodeExpression$, $o \in V_\aNodeExpression$, and the sub data-graph induced by $V_\aNodeExpression$ satisfies $\aNodeExpression$, where some of the data-values may have changed. Notice that there are at most $\binom{n}{c_{\aNodeExpression}} = \mathcal{O}(n^{c_{\aNodeExpression}})$ such subsets, which is polynomial on the input size.

As proposed in Theorem~\ref{theorem:repairTratable}, we could iterate on each of these subsets and consider every possible assignment of data values that only uses data values from $\Sigma_n^\aNodeExpression$ or an arbitrarily chosen set $A$ of data values that is not mentioned in $\aNodeExpression$ with $|A|=c_\aNodeExpression$. However, in this case, such set $A$ might not contain the optimum transitions based on the transitions costs defined by the function $\delta$. Nonetheless, we can use Lemma~\ref{lemma:optimumDatavalues} to show that there is a bounded number of assignments that we need to consider in order to find the optimum.
Therefore, we can adapt the algorithm of Theorem~\ref{theorem:repairTratable} to consider a specific set of data values assignments that we have shown contain an optimum assignment over $\delta$ and also satisfies $\aNodeExpression$ (if any such assignment exists). Observe that there are $|\Sigma_n^\aNodeExpression| + 2c_\aNodeExpression$ options for each node, thus there is a total of $(|\Sigma_n^\aNodeExpression| + 2c_\aNodeExpression)^{c_\aNodeExpression}$ assignments to consider. Finally, we keep the best assignment, which requires $O(n^{c_\aNodeExpression})$ operations.
\end{proof}


\section{Probabilistic query answering} \label{section:PQA}




In this section, we focus on the problem of probability query answering for PUDGs. We first define the problem
in general terms and later study the complexity for the case of subset and superset PUDG, respectively.

\begin{definition}\label{def:proba_de_nu_sobre_U}
Given $\aPUDG = (\aProbabilisticDatabase,\aRealizationModel,\aGraph')$ and a node or path expression $\aFormulaNodeOrPath$, we define \defstyle{the probability of $\aFormulaNodeOrPath$ over $\aPUDG$} (or over $\aProbabilisticDatabase$ given $\aGraph'$ and $\aRealizationModel$) as: 
$$\sum_{\aGraph \in \aProbabilisticDatabase} A_{\aGraph, \aFormulaNodeOrPath} \times  \aInverseRealizationModel(\aGraph')(\aGraph),$$
where $A_{\aGraph, \aFormulaNodeOrPath}$ is the binary value of ``$\aFormulaNodeOrPath$ holds over all nodes (or pairs of nodes) in $\aGraph$''.
\end{definition}

\begin{definition} [\functionProblem{Global PQA} over data-graphs] 
We define the probabilistic query answering problem for graph-databases: 
\begin{center} 
\fbox{\begin{minipage}{30em}
  \textsc{Problem}: The probability that a node [path] query $\aFormulaNodeOrPath$ holds over all nodes [pairs of nodes].  

\textsc{Input}: A PUDG $\aPUDG = (\aProbabilisticDatabase,\aRealizationModel,\aGraph')$

\textsc{Output}: The probability of $\aFormulaNodeOrPath$ over $\aProbabilisticDatabase$, given $\aGraph'$ and $\aRealizationModel$.
\end{minipage}}
\end{center}
\end{definition}

\begin{definition} [\functionProblem{Existential PQA} over data-graphs]
Similarly, we could define the probabilistic query answering for graph-databases: 
\begin{center} 
\fbox{\begin{minipage}{30em}
  \textsc{Problem}: The probability that a node [path] query $\aFormulaNodeOrPath$ holds over \emph{at least} one node [pair of nodes].  

\textsc{Input}: A PUDG $\aPUDG = (\aProbabilisticDatabase,\aRealizationModel,\aGraph')$

\textsc{Output}: The probability that $\aFormulaNodeOrPath$ holds over at least one node [pair of nodes] in the data-graphs from $\aProbabilisticDatabase$, given $\aGraph'$ and $\aRealizationModel$.
\end{minipage}}
\end{center}

However, notice that this problem is the dual of global PQA: given a node expression $\aNodeExpression$ [a path expression $\aPathExpression$], the answer to the Global PQA problem for that formula coincides with 1 minus the answer to the Existential PQA problem for $\lnot \aNodeExpression$ [$\complementoCamino{\aPathExpression}$], and vice versa:  

Let $A_{\aGraph, \aNodeExpression}$ be defined as in Definition~\ref{def:proba_de_nu_sobre_U}
for node expression $\aNodeExpression$, and $E_{\aGraph, \aNodeExpression}$ be the binary value of ``$\aNodeExpression$ holds over some node in $\aGraph$'', we have: 
\begin{align*}
\sum_{\aGraph \in \aProbabilisticDatabase} A_{\aGraph, \aNodeExpression} \times  \aInverseRealizationModel(\aGraph')(\aGraph)  &=  \sum_{\aGraph \in \aProbabilisticDatabase} (1 - E_{\aGraph, \lnot \aNodeExpression}) \times  \aInverseRealizationModel(\aGraph')(\aGraph) \\
&= \sum_{\aGraph \in \aProbabilisticDatabase} \aInverseRealizationModel(\aGraph')(\aGraph)  - \sum_{\aGraph \in \aProbabilisticDatabase} E_{\aGraph, \lnot \aNodeExpression} \times \aInverseRealizationModel(\aGraph')(\aGraph)  
\end{align*}

And therefore: 
$$ \sum_{\aGraph \in \aProbabilisticDatabase} A_{\aGraph, \aNodeExpression} \times  \aInverseRealizationModel(\aGraph')(\aGraph)   = 1 - \sum_{\aGraph \in \aProbabilisticDatabase} E_{\aGraph, \lnot \aNodeExpression} \times \aInverseRealizationModel(\aGraph')(\aGraph). $$

We also define a decision problem related to \functionProblem{Global PQA}, namely \decisionProblem{Global PQA bound}:

\begin{center} 
\fbox{\begin{minipage}{30em}
  \textsc{Problem}: Global PQA bound.
  
\textsc{Input}: A PUDG $\aPUDG = (\aProbabilisticDatabase,\aRealizationModel,\aGraph')$ and a bound $b$.

\textsc{Output}: Decide if the probability of $\aFormulaNodeOrPath$ over $\aProbabilisticDatabase$, given $\aGraph'$ and $\aRealizationModel$, is above $b$.
\end{minipage}}
\end{center}


\end{definition}

\input{Examples/example_globalPQAPathExpressionToySocialNetwork}

\begin{example} [Global PQA\textemdash probability that an (incomplete) observed network has no point of attack]
Suppose that an attacker observes a data-graph $\aGraph'$ representing an incomplete diagram of a computer network. Assume they want to query whether all the nodes of the (full) network are secure, assuming for simplicity that this is represented by them having the data value \esDato{secure}, and they have a model of possible networks ($\aProbabilisticDatabase$) and of the limitations of their data-gathering methods ($\aRealizationModel$).

In this case, it makes sense for the semantics of the probability in the PQA to represent the proportion of clean data-graphs where \emph{all} nodes are secure, i.e., the (weighted) proportion of $\aGraph \in \aProbabilisticDatabase$ such that for all $\aNode \in \aGraph$ we have that $\aGraph, \aNode \models \esDatoIgual{secure}$. The relative proportion of secure or not secure nodes should not be reflected in the calculated probability: the only thing that matters is the existence of at least a single point of attack. In other words, we are talking of the global PQA problem. 
\end{example}



\subsection{PQA for Subset PUDG}


As in Section~\ref{section:dataCleaning}, first we study the complexity of the PQA problem for subset PUDGs.

One way to solve the Global PQA problem is to evaluate the query $\aFormulaNodeOrPath$ in every graph in the cosupport of $\aGraph'$ and then sum all the results multiplied by the probabilities associated. If $\sigma(\aGraph')$ is polynomially bounded on $\aGraph'$, then this approach would be feasible, and therefore we conclude, based on our results for data cleaning, that:

\begin{theorem}

The Global PQA problem can be solved in polynomial time in Subset PUDGs if the number of nodes and edges deleted from the original graph, as well as the number of possible data-values that can be assigned to a node, are bounded by some constant.

\end{theorem}

Let us consider the simplified case in which no node deletions are allowed, and therefore the only possible data-graphs in the cosupport of $\aGraph'$ are graphs obtained by adding edges to $\aGraph'$. We can prove the following bounds:

\begin{theorem}
There exists a \Gregxpath formula $\aFormulaNodeOrPath$ such that the problem Global PQA is PP-hard when considering no node deletions through $\aRealizationModel$. Moreover, PSPACE is an upper bound for the problem.
\end{theorem}

\begin{proof}
For the first statement we will reduce \decisionProblem{MAJSAT} (\textit{majority SAT}) to subset PQA where no node deletions are allowed. \textsc{MAJSAT} problem consists in, given a boolean formula $\aFormula$ with free variables, deciding if more than half of the assignments of the free variables of $\aFormula$ satisfy it. We can assume that $\aFormula$ is given in conjunctive normal form.

Let $x_1,...,x_n$ be the free variables of $\aFormula$, and let $c_1,...,c_m$ be its clauses. We will construct an instance of the subset PQA problem such that half of the assignments of $\aFormula$ satisfy it if and only if the answer of the PQA problem is greater than $\frac{1}{2}$.

We define the observed data-graph $\aGraph'$ in the following way. Let $V_\aGraph' = \{v_i$ $|$ $1 \leq i \leq n\} \cup \{w_j$ $|$ $1 \leq j \leq m\} \cup \{\top, \bot\}$, where the nodes $v_i$ refer to the variables $x_i$ and the nodes $w_j$ refer to the clauses $c_j$. The edges of the graph are related to the structure of the formula $\aFormula$ as follows: 
$L(v_i,w_j) = \{\down_+\}$ if the variable $x_i$ appears in the clause $c_j$ without negation, and $L(v_i,w_j) = \{\down_-\}$ if the variable $x_i$ appears with negation in $c_j$. For every other pair of nodes $v,w$ that do not satisfy any of those hypotheses, we define $L(v,w)=\emptyset$. Also, we set $D(w_j)=clause$, $D(v_i)=var$ and $D(\star)=\star$ for $\star \in \{\top, \bot\}$. Observe that the data-graph $\aGraph'$ contains enough information to rebuild~$\aFormula$. 

The image of $\aInverseRealizationModel$ will represent the set of all possible assignments of the variables $x_1,...,x_n$. To achieve this, we define $\aProbabilisticDatabase(\aGraphb) = \frac{1}{2^n}$ if $\aGraph' \subseteq \aGraphb$, every node $v_i$ has an edge $\down_\esLabel{assigned}$ to either $\bot$ or $\top$ (but not to both at the same time), and if no other edge or node was added to obtain $\aGraphb$ from $\aGraph'$. We can easily map the data-graph $\aGraphb$ to an assignment of the variables $x_i$ by defining that $x_i$ is valuated to true if and only if the edge $(v_i, \esLabel{assigned}, \top)$ belongs to $\aGraphb$. We can also do this in the other way: given an assignment $v$ of the variables of $\aFormula$ we can build the data-graph $H_v$ that corresponds to that assignment.  

We also define $\aRealizationModel(\aGraph)(\aGraph')=1$ for every graph $\aGraph$. 

Finally, the query $\aFormulaNodeOrPath$ will be in charge of `checking' whether the assignment represented by any $\aGraphb \in \aInverseRealizationModel(\aGraph')$ satisfies the Boolean formula $\aFormula$. This can be captured by defining:

$$
\aFormulaNodeOrPath = 
[clause]^= 
\entoncesNodo 
\comparacionCaminos{\down^-_+ \down_\esLabel{assigned} [\top]^=} 
\lor
\comparacionCaminos{\down^-_- \down_\esLabel{assigned} [\bot]^=}
$$

Now we prove the following fact: \textit{a valuation $v$ over the variables of $\aFormula$ satisfies $\aFormula$ if and only if the data-graph $\aGraph'\subseteq H_v$ that represents such assignment $v$ satisfies the constraint $\aFormulaNodeOrPath$}.

Let $v$ be an assignment of the variables of $\aFormula$, and let $H_v$ be the data-graph that corresponds to that assignation. That is, $H_v = (V_{H_v}, L_{H_v}, D_{H_v})$ is defined as:
\begin{align*}
& V_{H_v} = V_{\aGraph'} \\
& D_{H_v}(z) = D_{\aGraph'}(z) \text{ for } z \in V_{H_v}\\
& L_{H_v}(v_i,\star) = \{\down_\esLabel{assigned}\} \text{ if } v(x_i) = \star \text{ for } \star \in \{\top, \bot\}. \text{ Otherwise } L_{H_v}(v_i,\star) = \emptyset \\
& L_{H_v}(z_1, z_2) = L_{\aGraph'}(z_1,z_2) \text{ for any other case}
\end{align*}

Suppose that $v$ satisfies $\aFormula$ and let $z$ be a node of $H_v$. If $D_{H_v}(z) \neq clause$ then $z \in \semantics{\aFormulaNodeOrPath}_{H_v}$ by definition of $\aFormulaNodeOrPath$. Otherwise, $z$ is actually $w_j$ for some $j$. Observe that the clause $c_j$ is satisfied by $v$, and therefore there is a variable $x_i$ that either appears without negation on $c_j$ and is valuated to true by $v$, or either appears negated in $c_j$ and is valuated to false by $v$. In the first case, there must be edges $(v_i,+,w_j)$ and $(v_i,\esLabel{assigned},\top)$ in the graph, while in the other case there will be edges $(v_i,-,w_j)$ and $(v_i,\esLabel{assigned},\bot)$. In either case, $c_j \in \semantics{\aFormulaNodeOrPath}_{H_v}$ since it satisfies the consequent of $\aFormulaNodeOrPath$.

Now for the other direction, suppose that $v$ does not satisfy $\aFormula$. Let $c_j$ be the clause that is not satisfied. Then, the node $w_j$ would not belong to $\semantics{\aFormulaNodeOrPath}_{H_v}$: if there is an edge $(v_i,+,c_j)$, then there will be no edge $(v_i,\esLabel{assigned},\top)$ in $H_v$, since that would imply that $v(x_i) = \top$ and $c_j$ is satisfied. The same reasoning can be followed for the $(v_i,-,c_j)$ edges. 

To complete the reduction, suppose that $\aFormula \in MAJSAT$. Thus, let $\mathcal{V}$ be the set of at least $2^{n-1} + 1$ assignments of $\aFormula$ that satisfy it. Observe that every data-graph of the set $\{H_v\}_{v \in \mathcal{V}}$ satisfies $\aFormulaNodeOrPath$, and also that:
$$
\aInverseRealizationModel(\aGraph')(H_v) = \frac{\aProbabilisticDatabase(H_v) \times \aRealizationModel(H_v)(\aGraph')}{\sum_{I \in \aProbabilisticDatabase} \aProbabilisticDatabase(I) \times \aRealizationModel(I)(\aGraph')} = \frac{\aProbabilisticDatabase(H_v)}{\sum_{I\in\aProbabilisticDatabase}\aProbabilisticDatabase(I)} = \aProbabilisticDatabase(H_v) = \frac{1}{2^n}.
$$

Then, the probability that $\aFormulaNodeOrPath$ holds over $\aProbabilisticDatabase$ given $\aGraph'$ and $\aRealizationModel$ must be equal or greater than $|\mathcal{V}| \times \frac{1}{2^n} \geq \frac{2^{n-1} + 1}{2^n} > \frac{1}{2}$.

For the other direction, if $\aFormula \notin MAJSAT$, then half or more of its assignments do not satisfy it. It follows that at least half of the data-graphs in $\aProbabilisticDatabase$ do not satisfy $\aFormulaNodeOrPath$, because there is a bijection between the data-graphs $H \in \aProbabilisticDatabase$ and the valuations over $\aFormula$, and this bijection relates formulas that satisfy $\aFormula$ with data-graphs that satisfy $\aFormulaNodeOrPath$. Since for every data-graph $H \in \aProbabilisticDatabase$ holds that $\aInverseRealizationModel(\aGraph')(H) = \frac{1}{2^n}$, it follows that the probability that $\aFormulaNodeOrPath$ holds over $\aProbabilisticDatabase$ given $\aGraph'$ and $\aRealizationModel$ is bounded by $\frac{1}{2}$.

\hspace{0.1cm}

For the upper bound, we show that the PQA subset problem can be solved in \textsc{PSPACE}. Notice that any data-graph $\aGraphb$ with the same number of nodes as the observed data-graph $\aGraph'$ can be represented by a binary string of length $n^2 \times |\Sigma_e|$ that defines which edges belong to $\aGraphb$. One way to solve the problem would be to iterate over this `counter' and, for each representation $\omega$, to construct the data-graph $\aGraphb$ and check whether the formula $\aFormulaNodeOrPath$ is satisfied in $\aGraphb$. We also use an accumulator where we sum the values $\aInverseRealizationModel(\aGraph')(\aGraphb)$ of those graphs that satisfy $\aFormulaNodeOrPath$. Since for every $\aGraphb$ the value $\aInverseRealizationModel(\aGraph')(\aGraphb)$ can be represented in polynomial size on $|\aGraph'|$, this accumulator also has polynomial size on $\aGraph'$\footnote{One way to prove this fact is by noticing that the number of bits that the counter needs to represent its current value will never be more than those needed by the value $\aInverseRealizationModel(\aGraph')(\aGraphb)$ with the biggest number of bits used. This number of bits can be bounded by a polynomial on $\aGraph'$, and therefore the number of bits used by the accumulator will always be bounded by this polynomial, since the sums will never require an extra bit to represent the output.}. All of this can be done in polynomial space, and requires exponential time.
\end{proof}

Observe that we could have used the node expression
\[
\aFormulab
=
[var]^= \vee [\top]^= \vee [\bot]^= 
\vee
\comparacionCaminos{\down^-_+ \down_\esLabel{assigned} [\top]^=} 
\lor
\comparacionCaminos{\down^-_- \down_\esLabel{assigned} [\bot]^=}
\]

to obtain the same bound but using an expression from \Gposregxpath, which is a weaker language.

\subsection{PQA for Superset PUDG}

We finish this section by giving complexity considerations for PQA in superset PUDGs.
In this case, the PQA problem can always be solved in polynomial space if $\aProbabilisticDatabase$ and $\aRealizationModel$ can be computed efficiently\footnote{In fact we only require them to use polynomial space to compute its outputs, which as well must be bounded by a polynomial}: a PSPACE algorithm would just iterate over every graph $\aGraph \subseteq \aGraph'$ and accumulate the value $A_{\aGraph_i,\aFormulaNodeOrPath} \times \aInverseRealizationModel(\aGraph')(\aGraph)$. Observe that the normalization factor of $\aInverseRealizationModel$ can also be computed. 

This is an improvement from the subset case: the general PQA subset problem was not solvable in PSPACE because we had no way of knowing whether the graphs generated while iterating over the cosupport of $\aGraph'$ were bounded by some polynomial. In this case, since $I$ is in the cosupport, we can guarantee that $I \subseteq \aGraph'$.

Nevertheless, the problem remains hard even in restricted scenarios:

\begin{theorem}

The PQA superset problem is \textsc{PP-hard} when considering no node additions through $\aRealizationModel$.

\end{theorem}

\begin{proof}

The proof is rather similar to the subset case. Given an input formula $\aFormula$ with $n$ variables $x_i$ and $m$ clauses $c_j$ for the MAJSAT problem, we define the graph $\aGraph' = (V, L, D)$ as follows:

\begin{align*}
& V  = \{v_i:1\leq i \leq n\} \cup \{w_j : 1\leq j \leq m\} \cup \{\top, \bot\} \\
& L(v_i,w_j)  = \{\down_+\} \iff x_i \mbox{ appears without negation in } c_j \\
& L(v_i,w_j)  = \{\down_-\} \iff x_i \mbox{ appears with negation in } c_j \\
& L(v_i,\top)  = \{\down_\esLabel{assigned}\} \mbox{ for every } 1\leq i \leq n \\
& L(v_i, \bot) = \{\down_\esLabel{assigned}\} \mbox{ for every } 1\leq i \leq n\\
& L(z_1,z_2)  = \emptyset \mbox{ for any other case} \\
& D(c_i)  = \esDato{clause} \\
& D(x_i)  = \esDato{var} \\
& D(\star)  = \star \mbox{ for } \star \in \{\top,\bot\} 
\end{align*}











We define $\aRealizationModel(I)(\aGraph') = 1$ for every $I\subseteq \aGraph'$, and $\aProbabilisticDatabase(I)=\frac{1}{2^n}$ if and only if $I$ represents an assignment of the variables $x_i$, following the semantics defined in the previous proof.

Using the same $\Gregxpath$ expression $\aFormulaNodeOrPath$ as in the subset case, it is easy to prove that $\aFormula$ is satisfied by more than half of its assignments if and only if the answer to superset PQA problem with inputs $\aGraph'$, $\aProbabilisticDatabase$, $\aRealizationModel$ and $\frac{1}{2}$ is 1.

\end{proof}
\section{Conclusions and future work} \label{Section:Conclusions}


In this work, we developed a formal framework for probabilistic unclean data-graphs, which allowed us to study several basic problems within them. The formalism defined in Section~\ref{Section:Definitions} expands the ideas described in~\cite{de2018formal}, addressing particular difficulties that arise when dealing with data-graphs. We also described a set of restrictions over the initial model that are analogous to those found in other works of the area, mainly the usual subset and superset restrictions considered when dealing with inconsistent databases as in~\cite{bienvenu2012complexity} or \cite{barcelo2017data}, and we show that they still hold a reasonable expressive power to capture real-life examples. We defined the problems of \decisionProblem{Data cleaning} and \decisionProblem{Probabilistic Query Answering} for unclean data-graphs, which strongly relate with two fundamental questions in the context of unclean databases developed in~\cite{de2018formal}. First, given an observed database and a model for the evolution of inconsistencies: which is the most likely original state of the data-graph? And second, given the same prior knowledge: which is the probability that a certain condition or constraint holds in the original state of the data-graph?

In Section~\ref{section:dataCleaning} and Section~\ref{section:PQA} we studied specific restrictions of the problems and determined that the complexity of the subset and superset versions of \decisionProblem{Data cleaning} and \decisionProblem{Probabilistic Query Answering} is hard in both cases. Nevertheless, we provided additional restrictions for which the problem becomes tractable. Regarding the study of \decisionProblem{PQA} in probabilistic data-graphs, we considered \decisionProblem{Global PQA} and \decisionProblem{Existential PQA} over data-graphs, and studied this using the \Gregxpath syntax to express conditions or constraints over data-graphs. It follows from the proofs given in Section~\ref{section:PQA} 
that the results obtained can be generalized to similar graph query languages, as long as they are (1) expressive enough to define simple path conditions on every pair of nodes and (2) simple enough to compute the set of answers given a data-graph instance in polynomial time.

Following a similar idea, we defined the \emph{update PUDG} restriction over the PUDG model, in which we consider the possibility of modification of labels and data values in edges and nodes, respectively. This resembles an analogous restriction already studied in the literature of inconsistency reasoning: namely, attribute-based repairs~\cite{Bertossi_2011_book}.
In the update case, we obtained intractable results for quite restrictive scenarios, which suggests that this problem is more difficult to address in the general case than those studied in Sections~\ref{section:dataCleaning} and~\ref{section:PQA}. 

We observed in Section~\ref{section:constraints} that the concept of soft and hard constraints from the area of database repairing can be adapted to our framework of EMDGs and PUDGs. We do so by defining (weighted) restriction sets, which modify the distribution of a probabilistic data-graph by reducing the probability of those data-graphs which violate the constraints in the restriction set.
A deeper complexity analysis for constraint satisfaction is a matter of our future work.

For the \decisionProblem{PQA} problem when considering only subset PUDGs, the best upper bound we show is \textsc{PSPACE}, while the lower bound is \textsc{PP}, so there might be some room for improvement. 
In addition, we did not find restrictions over the realization model $\aRealizationModel$ that could make the problems easy, without trivializing it. In most cases, the interaction between really simple probabilistic databases $\aProbabilisticDatabase$ and realization models $\aRealizationModel$ already makes the problems hard, so one possible way to advance would be to consider instances of the \decisionProblem{Data Cleaning} and \decisionProblem{PQA} problem where the distribution defined by $\aProbabilisticDatabase$ is uniform and the weight of the reasoning is encoded in $\aRealizationModel$.

Regarding the expressiveness of the query language, subsets of $\Gregxpath$ could be explored to lower the complexity of the \decisionProblem{PQA} problems. Nonetheless, observe that the results of hardness use only a narrow set of tools from the \Gregxpath grammar, and therefore serious limitations on the complexity of the expression should be imposed. Considering expressions from weaker navigation languages such as those mentioned in~\cite{barcelo2012expressive} might be a good place to start.
We also left open for exploration the possibility of semantic restrictions in the universe $\aUniverseOfGraphs$ of data-graphs, such as considering only data-trees, DAGs, or other type of structures, in order to investigate whether these restrictions provide a lowering of the complexity for our problems. Similarly, we can study the trade-offs involved in further restrictions on the probabilistic data-graphs $\aProbabilisticDatabase$ and the noisy observers $\aRealizationModel$. 
Finally, it would also be interesting to consider an alternative but natural definition of the semantics of a given formula $\aFormulaNodeOrPath$ over a data-graph, such that it returns the proportion of nodes or pairs of nodes that satisfy the formula.



\bibliographystyle{plain}
\bibliography{main}

\section{Appendix}\label{section:Appendix}

\begin{proof}[Proof of Theorem~\ref{teo-superset-fixed}]
This proof is analogous as the one provided for Theorem \ref{teo:subsetNpComplete}. The only difference is that we need to define $\aRealizationModel$ so that it actually removes the assignment of the data-graph by adding all possible edges of the form $(\star, \esLabel{value}, x_i)$ for $\star \in \{\bot, \top\}$ and $1 \leq i \leq n$. Such $\aRealizationModel$ can be defined as
$$
\aRealizationModel(\aGraph)(H) = 2^{-(|\Sigma_e||V_\aGraph)|^2 - |E_\aGraph|)}
$$
if $V_\aGraph = V_H$ and $\aGraph \subseteq H$. Otherwise, $\aRealizationModel(\aGraph)(H) = 0$.

Now, for the reduction, given a 3CNF formula $\phi$ we map it to the data-graph $\aGraph_{\phi, f}$ where $f$ intuitively ``assigns'' both boolean values to each variable. Deciding if there is some data-graph such that $\aInverseRealizationModel_{\aProbabilisticDatabase, \aRealizationModel} > \frac{1}{2}^n$ is equivalent to deciding if there exists a satisfying assignment for $\phi$.
\end{proof}

\begin{proof}[Proof of Theorem~\ref{teo-update-fixed}]
Even though the main ideas are the same as in Theorem~\ref{teo:subsetNpComplete}, for this proof we have to define a somewhat different $\aRealizationModel$ and reduction. First, $\aRealizationModel$ will once again delete the assignment represented by the data-graph, but in this case this information needs to be represented through the edge labels. That is, we consider that a data-graph represents a 3CNF formula $\phi$ alongside an assignment $f$ of its variables if we have the $m$ clause and $n$ variable nodes with the corresponding \esLabel{is\_literal} and \esLabel{is\_literal\_negated}, the two boolean nodes $\bot$ and $\top$ with an edge between them, but now the assignment $f$ will be present in the data-graph in the following sense: each variable node will have an edge from itself to $\bot$ and $\top$, but only one of those will have the label \esLabel{chosen}, while the other one will have the label \esLabel{unchosen}. Then, we can represent the ``deletion'' of an assignment satisfying $\phi$ by turning all \esLabel{chosen} edges to \esLabel{unchosen} ones. Observe that the probabilistic database $\aProbabilisticDatabase$ can be defined in an analogous way. Finally, the realization model $\aRealizationModel$ is defined as:
$$
\aRealizationModel(\aGraph)(H) = 2^{-|E_\aGraph^{\lnot \esLabel{unchosen}}|}
$$
if $V_\aGraph = V_H$, $\aGraph$ has the same edges as $H$ up to a renaming of the edge labels and the only edges with different labels have a label \esLabel{unchosen} in $H$. In this formula, $E_\aGraph^{\lnot \esLabel{unchosen}}$ denotes the set of edges from $\aGraph$ with label different from \esLabel{unchosen}.

As in the proof of Theorem~\ref{teo:subsetNpComplete}, we complete the reduction from 3SAT: given a $3CNF$ formula $\phi$ we build the data-graph $\aGraph'=\aGraph_{\phi,f}$, where $f$ assigns to each variable $x_i$ both boolean values. Moreover, we consider $b = \frac{1}{2^n}$.
\end{proof}

\begin{proof}[Proof of Lemma~\ref{lema:constraints1}]
Let $\eta$ be a $\Gregxpath$ formula, $\aGraph=(V_\aGraph, L_\aGraph, D_\aGraph)$ a data-graph of $n$ nodes and $A$ a set of data values such that $|A| = n$ and $\Sigma_n^{\eta} \cap A = \emptyset$. Let $C = \{\esDato{c}_1, \esDato{c}_2, \ldots, \esDato{c}_k\}$ be the set of all data values present in $\aGraph$ that are not in $\Sigma_n^\eta$, and let $\aGraph'=(V_\aGraph, L_\aGraph, D_{\aGraph'})$ be a new data-graph such that
\[
D_{\aGraph'}(v) = \begin{cases} 
          \esDato{d}_i & D_\aGraph(v)= \esDato{c}_i \text{ for some } i\\
         D_\aGraph(v) & otherwise\\
       \end{cases}
\]

We prove by induction on the formula's structure that:

\begin{enumerate}
    \item If $\eta$ is a path expression, then $(v,w) \in \semantics{\eta}_\aGraph \iff (v,w) \in \semantics{\eta}_{\aGraph'}$.
    \item If $\eta$ is a node expression, then $v \in \semantics{\eta}_\aGraph \iff v \in \semantics{\eta}_{\aGraph'}$
\end{enumerate}

The only interesting cases are those in which the expression interacts with data values, since the edges of the data-graph were not modified. Observe that the theorem is true for the base cases: if $\eta = [\esDatoIgual{c}]$, then clearly $\semantics{\eta}_\aGraph = \semantics{\eta}_{\aGraph'}$, since those nodes with data values in $\Sigma_n^\eta$ kept their original data. Now, if $\eta = [\esDatoDistinto{c}]$ the theorem also holds: if $v \in \semantics{\esDatoDistinto{c}}_\aGraph$ then $D_\aGraph(v)$ is in $\Sigma_n^\eta \setminus \{\esDato{c}\}$ and was not modified in $\aGraph'$, or rather it is not in $\Sigma_n^\eta$, and belongs to $A$. Since $A \cap \Sigma_n^\eta = \emptyset$ we conclude that $v \in \semantics{\esDatoDistinto{c}}_{\aGraph'}$.

For the inductive cases, we consider:

\begin{itemize}
    \item $\eta \equiv \langle \aPath = \aPathb \rangle$: this case follows easily by noticing that, if two nodes had the same data value in $\aGraph$ then they also have the same data value in $\aGraph'$.
    \item $\eta \equiv \langle \aPath \neq \aPathb \rangle$: this case follows by observing that, if two nodes had a different data value in $\aGraph$, then they also contain a different data value in $\aGraph'$.
\end{itemize}
\end{proof}

\begin{proof}[Proof of Lemma~\ref{lemma:optimumDatavalues}]
Let $\aNodeExpression \in \Gposcoreregxpath$ be a node expression that does not use the $\langle \aPath = \aPathb \rangle$ subexpression, $\aGraph = (V_\aGraph, L_\aGraph, D_\aGraph)$ a data-graph with an origin $o$ and $\delta:\Sigma_n \times \Sigma_n \to \N_0$ a function such that for each $c \in \Sigma_n$ it is possible to efficiently enumerate $\esDato{d}_1, \esDato{d}_2, \ldots$ where $\delta(\esDato{d}_i, c) \leq \delta(\esDato{d}_{i+1},c)$ and every data value $\esDato{d}$ belongs to the enumeration.

We want to show that, if there exists a data-graph $H=(V_\aGraph, L_\aGraph, D_H)$ such that $H, o \models \aNodeExpression$ and $\sum_{v\in V_\aGraph} \delta(D_H(v), D_\aGraph(v))$ is minimized across all data-graphs equal to $\aGraph$ up to its data values that also satisfy $\aNodeExpression$ at the origin, then there is another data-graph $H'=(V_\aGraph, L_\aGraph, D_{H'})$ such that $H', o \models \aNodeExpression$, $H'$ also minimizes $\sum_{v\in V_\aGraph} \delta(D_{H'}(v), D_\aGraph(v))$ and for each data value $\esDato{c}$ in $\aGraph$ that was changed into $\esDato{d}$ in $H'$ where $\esDato{d}$ is not mentioned by $\aNodeExpression$ it is satisfied that $\esDato{d}=\esDato{d}_j$ for $j \leq n + |\Sigma_n^\aNodeExpression|$ where $\esDato{d}_i$ is the $i$th data value in the enumeration $\esDato{d}_1, \esDato{d}_2, \ldots$ such that $\delta(\esDato{d}_i, \esDato{c}) \leq \delta(\esDato{d}_{i+1}, \esDato{c})$.

To prove this, it suffices to find, given $H$ and a node $v$ such that $D_H(v) \notin \Sigma_n^\aNodeExpression$ and $D_H(v) \neq \esDato{d}_j$ for all $j\leq n + |\Sigma_n^\aNodeExpression|$ where the $\esDato{d}_j$ are the numeration of $\delta$ for $D_\aGraph(v)$, a new data value $\esDato{d}$ such that $\esDato{d} = \esDato{d}_j$ for $j \leq n + |\Sigma_n^\aNodeExpression|$, the data-graph obtained by considering $H$ and modifying the data value of $v$ to $\esDato{d}$ also satisfies $\aNodeExpression$ when evaluated at the origin, and the summation involving $\delta$ is still minimized.

Let $\Sigma_n^H$ be the set of data values that are contained by some node of $H$. Also, observe that the set $A = \{\esDato{d}_j : 1 \leq j \le n + |\Sigma_n^\aNodeExpression|, \esDato{d}_j \notin \Sigma_n^\aNodeExpression\}$ has size at least $n$. Then, since $D_H(v) \notin A$ there is at least one data value $\esDato{d} \in A \setminus \Sigma_n^H$. We claim that the data-graph $H' = (V_H, L_H, D_{H'})$ where $D_{H'}(w) = D_H(w)$ unless $w =v$, in which case $D_{H'}(v) = \esDato{d}$, satisfies all the desired conditions.

Clearly, $\delta$ is also minimized in $H'$, since we only chose a transition with smaller cost. Using the following claim we conclude that $H'$ satisfies $\aNodeExpression$ at the origin.
\end{proof}

\textit{Claim}: Let $G$ be a data-graph and $\eta$ an expression from $\Gposregxpath$ without data comparison with equality (i.e. the subexpression $\langle \aPath = \aPathb \rangle$). Then, given a node $v$ such that $D_\aGraph(v) \notin \Sigma_n^\aNodeExpression$, we can define the data-graph $\aGraph'=(V_\aGraph, L_\aGraph, D_{\aGraph'})$ where $D_{\aGraph'}(w) = D_\aGraph(w)$ if $w = v$, and $D_{\aGraph'} = \esDato{d} \in \Sigma_n \setminus (\Sigma_n^\aNodeExpression \cup \Sigma_n^\aGraph)$. Then, $\semantics{\eta}_\aGraph = \semantics{\eta}_{\aGraph'}$.

\begin{proof}[Proof of the Claim:]

We prove this by structural induction. As before, the only interesting cases are those in which the subexpression interacts with data values. For the base cases:

\begin{itemize}
    \item $\eta \equiv \esDato{c}^=$: if $w \in \semantics{\esDato{c}^=}_\aGraph$ we can conclude that $w \in \Sigma_n^\aNodeExpression$, which implies that $w \neq v$. Then, $D_{\aGraph'}(w) = D_\aGraph(w)$, and $w \in \semantics{\esDato{c}^=}_{\aGraph'}$. The other direction can be seen in the same way.
    \item $\eta \equiv \esDato{c}^\neq$: if $w \in \semantics{\esDato{c}^\neq}$, then $D_{\aGraph}(w) \notin \Sigma_n^{\aNodeExpression}$, and therefore $D_{\aGraph'}(w) \notin \Sigma_n^\aNodeExpression$ and $w \in \semantics{\esDato{c}^\neq}_{\aGraph'}$. The other direction follows with similar arguments.
\end{itemize}

For the recursive cases, we only need to consider $\eta \equiv \langle \aPath_1 \neq \aPath_2 \rangle$. Observe that if $w \in \semantics{\eta}_{\aGraph}$, then there exists nodes $z_1,z_2$ such that $(v,z_i) \in \semantics{\aPath_i}_{\aGraph}$ and $D(z_1)_\aGraph \neq D(z_2) \aGraph$. We know that $(v, z_i) \in \semantics{\aPath_i}_{\aGraph'}$ by induction, so we can conclude by proving that $D_{\aGraph'}(z_1) \neq D_{\aGraph'}(z_2)$. Observe that this inequality could only have been altered if $z_i=v$ for some $i$. But, in that case, the inequality still holds, since $D_{\aGraph'}(v) \notin \Sigma_n^H$ by construction.
The other direction can be proven in the exact same way.
\end{proof}


\end{document}